\pdfoutput=1
\documentclass[a4paper,11pt]{amsart}

\newcommand{\mc}[1]{\mathcal{#1}}
\newcommand{\mb}[1]{\mathbf{#1}}
\newcommand{\mbb}[1]{\mathbb{#1}}
\newcommand{\T}{\mbb{T}}
\newcommand{\I}{\mbb{I}}

\newcommand{\mr}[1]{\mathrm{#1}}

\newcommand{\ms}[1]{\mathsf{#1}}

\newcommand{\Set}{\mb{Set}}

\newcommand{\Pos}{\mb{Pos}}
\newcommand{\alg}{\text{-}\mb{Alg}}

\newcommand{\DL}{\mb{DL}}

\newcommand{\sh}{\mb{Sh}}
\newcommand{\psh}{\mb{Psh}}

\newcommand{\op}{^{\mathrm{op}}}

\newcommand{\qsi}[1]{\tilde{#1}}

\newcommand{\ov}[1]{\overline{#1}}
\newcommand{\set}[1]{\{\,#1\,\}}

\newcommand{\pair}[1]{\left\langle#1\right\rangle}

\newcommand{\ev}{\mathrm{ev}}

\newcommand{\nt}{\Rightarrow}
\newcommand{\scomp}[2]{\{\,#1\mid#2\,\}}

\newcommand{\surj}{\twoheadrightarrow}
\newcommand{\hook}{\hookrightarrow}

\newcommand{\prth}[1]{\left(#1\right)}

\newcommand{\fp}{_{\mr{f.p.}}}

\newcommand{\cp}{_{\mr{c.p.}}}

\newcommand{\cv}{\operatorname{\downarrow}}

\newcommand{\N}{\mb N}

\newcommand{\wCPO}{\omega\mb{CPO}}

\newcommand{\dneg}{\neg\neg}

\newcommand{\fa}[2]{\forall #1\!\colon\!\!#2\mathpunct{.}}
\newcommand{\ex}[2]{\exists #1\!\colon\!\!#2\mathpunct{.}}

\newcommand{\ld}[2]{\lambda #1\!\colon\!\!#2\mathpunct{.}}

\newcommand{\emp}{\emptyset}
\newcommand{\eq}{\leftrightarrow}
\newcommand{\pss}[1]{\lVert #1\rVert} 
\newcommand{\pt}{\ms{pt}}

\newcommand{\pp}{\ms{Prop}}
\newcommand{\stt}{\ms{Set}}

\newcommand{\sFrm}{\sigma\mb{Frm}}
\newcommand{\Frm}{\mb{Frm}}
\newcommand{\Loc}{\mb{Loc}}

\newcommand{\Topp}{\mb{Esp}}

\newcommand{\wTop}{\omega\mb{Esp}}

\newcommand{\mmod}[1]{#1\text{-}\mathbf{Mod}}

\newcommand{\spec}{\operatorname{Spec}}

\newcommand\istsym{\ms{t}}
\newcommand\isfsym{\ms{f}}
\newcommand\ist[1]{\istsym(#1)}
\newcommand\isf[1]{\isfsym(#1)}

\NewDocumentCommand\obsle{o}{\sqsubseteq_{\ms{obs}}\IfValueT{#1}{^{#1}}}
\NewDocumentCommand\satle{o}{\le_{\ms{sat}}\IfValueT{#1}{^{#1}}}
\NewDocumentCommand\behle{o}{\sqsubseteq_{\ms{beh}}\IfValueT{#1}{^{#1}}}
\NewDocumentCommand\opens{}{\mc{O}}

\makeatletter
\newcommand{\ct@}[2]{%
  \vtop{\m@th\ialign{##\cr
    \hfil$#1\operator@font lim$\hfil\cr
    \noalign{\nointerlineskip\kern1.5\ex@}#2\cr
    \noalign{\nointerlineskip\kern-\ex@}\cr}}%
}
\newcommand{\ct}{%
  \mathop{\mathpalette\ct@{\rightarrowfill@\textstyle}}\nmlimits@
}
\makeatother
\makeatletter
\newcommand{\lt@}[2]{%
  \vtop{\m@th\ialign{##\cr
    \hfil$#1\operator@font lim$\hfil\cr
    \noalign{\nointerlineskip\kern1.5\ex@}#2\cr
    \noalign{\nointerlineskip\kern-\ex@}\cr}}%
}
\newcommand{\lt}{%
  \mathop{\mathpalette\lt@{\leftarrowfill@\textstyle}}\nmlimits@
}
\makeatother

\NewDocumentCommand\AxiomSQCI{}{SQCI}
\NewDocumentCommand\PrintAxiomSQCI{}{
  \begin{axiom}[\AxiomSQCI]
    $\I$ is stably spatial, i.e.\ the slice $\I/i$ is spatial for every $i:\I$.
  \end{axiom}
}

\NewDocumentCommand\AxiomSQCID{}{SQCID}
\NewDocumentCommand\PrintAxiomSQCID{}{
  \begin{axiom}[\AxiomSQCID]
    Both $\I$ and $\I\op$ are stably spatial.
  \end{axiom}
}

\NewDocumentCommand\AxiomSQCP{}{SQCP}
\NewDocumentCommand\PrintAxiomSQCP{}{
  \begin{axiom}[\AxiomSQCP]
    The polynomial $\I$-algebra $\I[\ms{i}]$ is spatial.
  \end{axiom}
}

\NewDocumentCommand\AxiomNT{}{NT}
\NewDocumentCommand\PrintAxiomNT{}{
  \begin{axiom}[\AxiomNT]
    We have $0 \neq 1$ in $\I$.
  \end{axiom}
}

\NewDocumentCommand\AxiomL{}{L}
\NewDocumentCommand\PrintAxiomL{}{
  \begin{axiom}[\AxiomL]
    $\I$ is local, i.e.\ $0 \neq 1$, and $\ist{i\vee j} \eq \ist{i} \vee \ist{j}$ for $i,j : \I$. 
  \end{axiom}
}

\NewDocumentCommand\AxiomCL{}{cL}
\NewDocumentCommand\PrintAxiomCL{}{
  \begin{axiom}[\AxiomCL]
    $\I$ is colocal, i.e.\ $0 \neq 1$ and $\isf{i\wedge j} \eq \isf{i} \vee \isf{j}$ for all $i,j : \I$.
  \end{axiom}
}

\NewDocumentCommand\AxiomSL{}{SL}
\NewDocumentCommand\PrintAxiomSL{}{
  \begin{axiom}[\AxiomSL]
    $\I$ is a \emph{strict linear order}, i.e.\ $0 \neq 1$ and $i \le j \vee j \le i$ for all $i,j:\I$. 
  \end{axiom}
}

\NewDocumentCommand\AxiomOneCS{}{1cS}
\NewDocumentCommand\PrintAxiomOneCS{}{
  \begin{axiom}[\AxiomOneCS]
    $\I$ is \emph{1-coskeletal}, i.e.\ (\AxiomSL) holds and $i \le j \eq \isf{i} \vee (i = j)\vee \ist{j}$ for all $i,j : \I$.
  \end{axiom}
}

\NewDocumentCommand\AxiomSQCF{}{SQCF}
\NewDocumentCommand\PrintAxiomSQCF{}{
  \begin{axiom}[\AxiomSQCF]
    All finitely presented $\I$-algebras are spatial.
  \end{axiom}
}

\NewDocumentCommand\AxiomSQCC{}{SQCC}
\NewDocumentCommand\PrintAxiomSQCC{}{
  \begin{axiom}[\AxiomSQCC]
    All countably presented $\I$-algebras are spatial.
  \end{axiom}
}

\usepackage{amsmath,amssymb,amsthm,amsfonts}
\usepackage[numbers]{natbib}
\usepackage[english]{babel}
\usepackage[X2,T1]{fontenc}
\usepackage[tt=false]{libertine}
\usepackage{libertinust1math}
\usepackage[utf8]{inputenc} 
\usepackage{enumitem} 
\usepackage{tikz-cd}
\usepackage{hyperref}
\usepackage{stmaryrd}
\usepackage{appendix}
\usepackage{mathtools}
\usepackage{zref-clever}
\zcsetup{cap=true}
\usepackage{quiver}
\usepackage{microtype}

\newtheorem{theorem}{Theorem}[section]
\newtheorem{lemma}[theorem]{Lemma}
\newtheorem{corollary}[theorem]{Corollary}
\newtheorem{proposition}[theorem]{Proposition}

\theoremstyle{definition}
\newtheorem{example}[theorem]{Example}
\newtheorem{convention}[theorem]{Convention}
\newtheorem{definition}[theorem]{Definition}
\newtheorem{remark}[theorem]{Remark}

\newtheorem{observation}[theorem]{Observation}
\newtheorem*{axiom}{Axiom}

\zcLanguageSetup{english}{
  type = observation,
    Name-sg = Observation,
    name-sg = observation,
    Name-pl = Observations,
    name-pl = observations,
}

\AddToHook{env/lemma/begin}{\zcsetup{countertype={theorem=lemma}}}
\AddToHook{env/corollary/begin}{\zcsetup{countertype={theorem=corollary}}}
\AddToHook{env/proposition/begin}{\zcsetup{countertype={theorem=proposition}}}
\AddToHook{env/definition/begin}{\zcsetup{countertype={theorem=definition}}}
\AddToHook{env/example/begin}{\zcsetup{countertype={theorem=example}}}
\AddToHook{env/remark/begin}{\zcsetup{countertype={theorem=remark}}}
\AddToHook{env/notation/begin}{\zcsetup{countertype={theorem=notation}}}
\AddToHook{env/convention/begin}{\zcsetup{countertype={theorem=convention}}}
\AddToHook{env/observation/begin}{\zcsetup{countertype={theorem=observation}}}
\AddToHook{env/axiom/begin}{\zcsetup{countertype={theorem=axiom}}}

\title{Domains and classifying topoi}
\author{Jonathan Sterling}
\address{%
Jonathan \textsc{Sterling}\newline
Department of Computer Science and Technology\newline
University of Cambridge\newline
Cambridge, UK%
}

\author{Lingyuan Ye}
\address{%
Lingyuan \textsc{Ye}\newline
Department of Computer Science and Technology\newline
University of Cambridge\newline
Cambridge, UK%
}

\begin{document}

\begin{abstract}
  We explore a new connection between synthetic domain theory and Grothendieck topoi related to the distributive lattice classifier. In particular, all the axioms of synthetic domain theory (including the inductive fixed point object and the chain completeness of the dominance) emanate from a countable version of the \emph{synthetic quasi-coherence} principle that has emerged as a central feature in the unification of synthetic algebraic geometry, synthetic Stone duality, and synthetic category theory. The duality between spatial algebras and sober spaces in a topos with a distributive lattice object provides a new set of techniques for reasoning synthetically about domain-like structures, and reveals a broad class of (higher) sheaf models for synthetic domain theory. 
\end{abstract}

\maketitle
\tableofcontents

\section{Introduction}\label{sec:intro}

\subsection{Synthetic domain theory}\label{subsec:sdt}

The proposal to develop a synthetic theory of domains was raised by Dana Scott in the early 1980s. The thesis is that (pre)domains should be viewed simply as some special sets in a suitable universe, and any function between them should automatically respect the computational data associated to domains. Later \citet{hyland1990first} gave an extensive list of the properties such a universe might satisfy.

Firstly, a universe for synthetic domains should contain an interval object $\I$, which induces an observational preorder (or \emph{specialisation} preorder from a topological perspective) on domains for which function is automatically monotone. An internal axiom that characterises the synthetic nature of this fact is the so-called \emph{Phoa principle}: The function space $\I^\I$ should classify the order on $\I$. In other words, the functions from $\I$ to itself are completely determined by the order structure on $\I$. 

Additionally, $\I$ should also form a subuniverse of propositions closed under dependent sums, so as to form a dominance in the sense of \citet{rosolini1986continuity}. The dominance structure is used to parametrise partial functions through a \emph{partial map classifier} functor $L$ constructed from $\I$. 

Besides these elementary axioms, the main additional axiom of synthetic domain theory is the \emph{chain completeness} of the interval: 
\begin{itemize}
  \item[]
  Let $\omega$ and $\ov\omega$ be the carriers of the internal initial algebra and final coalgebra to the partial map classifier functor $L$ respectively. There is a canonical inclusion $\omega \hook \ov\omega$, and the chain completeness axiom states that $\I$ is \emph{right orthogonal} to this inclusion in the sense that the canonically induced restriction map $\I^{\ov\omega} \to \I^{\omega}$ should be an equivalence. 
\end{itemize} 
This can be viewed as the synthetic version of the $\omega$-completeness of the Sierpi\'nski space in traditional domain theory.

\subsection{Two viewpoints on synthetic domains}

The most studied models of synthetic domain theory are those arising from realisability---in which the interval $\I$ is taken to be Rosolini's dominance of semidecidable propositions in the Effective Topos~\citep{hyland1990first,PhoaWesleyKym-Son1991DtiR}, which happens to coincide there with the \emph{initial $\sigma$-frame} by virtue of countable choice. The model of synthetic domain theory in $\mathbf{Eff}$ ultimately interprets domain theoretic concepts in terms of Turing machines, which makes them applicable to the classical theory of computation~\citep{RN552}. 

On the other hand, realisability models do not provide a direct connection to classical theories of domains. For the latter, researchers returned to the original suggestion of Scott and devised sheaf models of synthetic domain theory in which the site is obtained from some existing category of domains or a dense generator thereof~\citep{fiore-plotkin:1996,FIORE1997151,fiore2001domains}.

Much effort has been put toward identifying a common axiomatics that accounts for both realisability and sheaf models of synthetic domain theory; cf.\ the work of \citet{reus-streicher:1999,simpson:2004}. Beyond these basic axioms, the realisability and sheaf models of synthetic domain theory behave very differently, as exemplified in the failure of the initial $L$-algebra to be inductive in the former, as \citet{VANOOSTEN2000233} pointed out.

The goal of the present paper is to identify a number of axioms emerging from the theory of \emph{classifying topoi} that suffice to explain the behaviour of sheaf models of synthetic domain theory all at once, including not only the dominance property and Phoa's principle, but also the chain completeness of the interval and the inductive construction of the initial $L$-algebra. We stress that the strongest of these axioms \emph{cannot} hold in the realisability model of synthetic domain theory, but they can be used to generate \emph{new} and exotic sheaf models of synthetic domains.

\subsection{Classifying topoi and Phoa's principle}\label{subsec:classtopphoa}

As a preliminary observation, there is a connection between Phoa's principle and classifying topoi (cf.\ \citet[Lem 3.8]{gratzer2024directed}).
\emph{A posteriori}, the technical elements required for this observation are already contained in the work of \citet{RN879}.

To recap, let $\T$ be a (finitary) Horn theory. Recall that the classifying topos of $\T$ is given by the following presheaf category (cf.\ \citet[D3.1]{johnstone2002sketches}),
\[ \Set[\T] \coloneq [\mmod\T\fp,\Set]\text{,} \]
where $\mmod\T\fp$ is the category of finitely presented $\T$-algebras. In particular, there is a generic $\T$-model $U_\T$, such that any $\T$-model in a topos $\mc E$ is realised as the inverse image of $U_\T$ under an essentially unique geometric morphism $\mc E \to \Set[\T]$. The observation of \citet{RN879} is that, internally in $\Set[\T]$, the function space of $U_\T$ is completely characterised by polynomials,
\[ U_\T^{U_\T} \cong U_\T[\ms{x}]\text{,} \]
where $U_\T[\ms{x}]$ denotes the free $\T$-model generated by an additional element, which is a  polynomial algebra on $U_\T$.

\citet{gratzer2024directed} have shown that when $\T$ is taken to be the theory $\mbb D$ of \emph{bounded distributive lattices} and $\I$ is the generic $\mbb D$-model, the equivalence $\I^\I \cong \I[\ms{i}]$ implies Phoa's principle---because for bounded distributive lattices, the polynomial $\I$-algebra $\I[\ms{i}]$ always classifies the partial order on $\I$. This is a striking example of how the properties of a theory impact the internal logic of its classifying topos. We will see many more examples in this paper.

\subsection{Synthetic quasi-coherence}\label{subsec:qc}

Recently, the work of Blass has been greatly generalised by Blechschmidt in his PhD thesis~\citep{blechschmidt2021using} (also see the unpublished note~\citep{blechschmidt2020general} by the same author), which identifies a stronger property satisfied by the generic model $U_\T$ in $\Set[\T]$, termed \emph{(synthetic) quasi-coherence}. The terminology is motivated by its application in algebraic geometry.

The main algebraic objects of study within the classifying topos $\Set[\T]$ will be the $U_\T$-algebras, i.e.\ $\T$-models $A$ equipped with a homomorphism $U_\T\to A$. Of course, the identity homomorphism exhibits $U_\T$ as a $U_\T$-algebra, and given any object $X$, we may define \[\opens X\coloneq U_\T^X\] to be the $X$-fold product $\prod_{x:X}U_\T$ of $U_\T$-algebras. 

Conversely, given any $U_\T$-algebra $A$ we may define its \emph{spectrum} to be the type of $U_\T$-algebra homomorphisms 
\[ \spec A \coloneq U_\T\alg(A,U_\T)\text{.}\]

The quasi-coherence principle in $\Set[\T]$ proposed by \citet{blechschmidt2021using} is, then, that for any \emph{finitely presented} $U_\T$-algebra $A$, the canonical homomorphism
\[
  A\to \opens\spec A
\]
obtained on elements by transposing the evaluation map $A\times \spec A\to U_\T$ is an equivalence.
Types of the form $\spec A$ for such $A$ are called `affine' by \citet{blechschmidt2021using}. This in fact induces an internal duality between finitely presented $U_\T$-algebras and `affine' types. In particular, it implies Blass's result. 

The appearance of finite presentation in the formulation of quasi-coherence above is perhaps misleading. Following the development of \citet{blechschmidt2021using,blechschmidt2020general}, it is clear that the quasi-coherence principle with respect to finitely presented algebras can be generalised to larger cardinality,
if we enlarge the site of the classifying topos $\Set[\T]$. For instance, one can consider
\[ \Set[\T]_\omega \coloneq [\mmod\T\cp,\Set]\text{,} \]
where $\mmod\T\cp$ denotes the category of \emph{countably presented} $\T$-models. In this case, the quasi-coherence principle will hold also with respect to \emph{countably presented} $U_\T$-algebras. This is the theoretical basis in a recent approach for synthetic topology suggested by \citet{cherubini2024foundation}, who consider a variation on the quasi-coherence principle for \emph{countably} presented Boolean algebras. 
Horn theories $\T$ whose operations have \emph{countable} arity can also be considered in this case, so long as the base topos $\Set$ satisfies countable choice. These examples will not fit directly into the discipline of classifying toposes because geometric theories cannot have infinitary operations: but a suitable doctrine may arise when considering geometric morphisms whose inverse image functors preserve countable rather than finite limits.

Because of the importance of varying the class of algebras with respect to which the quasi-coherence principle holds, we shall take a slightly different approach to the axiomatics inspired by locale theory and synthetic topology. In particular, the homomorphism $A\to \opens\spec A$ is the \emph{counit} of a contravariant adjunction $\opens\dashv \spec\colon U_\T\alg\op\to \Set$, so the algebras with respect to which the quasi-coherence principle holds are \emph{precisely} those that are fixed points of this adjunction.  We will therefore refer to algebras and sets as \emph{spatial} and \emph{sober} respectively when they arise as fixed points of this adjunction. Therefore, sober sets play the role of `affine schemes' in our setting, and a ``quasi-coherence principle'' is then no more than an assumption that certain classes of algebras are spatial.

Blechschmidt also points out that quasi-coherence descends from the classifying (presheaf) topos $\Set[\T]$ to any subtopos containing the generic model $U_\T$; see \citet[Cor. 7.7]{blechschmidt2020general}. From a logical perspective, a subtopos $\Set[\mbb S] \hook \Set[\T]$ can be identified as the classifying topos for a \emph{geometric quotient} $\mbb S$ of the Horn theory $\T$. If $U_\T$ is a sheaf for $\Set[\mbb S]$, then $U_\T$ in the subtopos $\Set[\mbb S]$ will be the generic $\mbb S$-model.
Hence, in such a subtopos, $U_\T$ will validates more geometric sequents as specified by the topology, which should be viewed as certain \emph{local properties}. For instance, \citet{blechschmidt2021using,Cherubini_Coquand_Hutzler_2024} work with \emph{local rings}, where the additional geometric properties are validated by the Zariski topology.
In the case of synthetic Stone duality, the Boolean algebra used by \citet{cherubini2024foundation} is the set 2, which is again only true in a suitably chosen topology.

\subsection{Summary of contributions}\label{sec:contributions}

We study an emerging connection between synthetic domain theory and the quasi-coherence axioms for \emph{bounded distributive lattices} and \emph{$\sigma$-frames}. As a first indication of why quasi-coherence for these theories might be useful for domain theory, we have observed that by taking the interval $\I$ to be a bounded distributive lattice or $\sigma$-frame internal to some topos,
many important objects for domain theory can be written as \emph{spectra}. This include the cubes $\I^n$ (\zcref{exm:cubesober}), the simplices $\Delta^n$ (\zcref{exm:simplicessober}), and the final coalgebra $\ov\omega$ (\zcref{exm:ovomegasober}). Here we outline the main results of this paper, and discuss some of their potential applications.

\begin{enumerate}[leftmargin=*]
  \item \emph{Quasi-coherence produces synthetic domain theory.} We show that several quasi-coherence principles suffice to account for \emph{all} essential axioms for synthetic domain theory. Specifically:
  \begin{enumerate}
    \item\label{contribution:dominance} In the case of both bounded distributive lattices and $\sigma$-frames, the \emph{dominance property} is identified as a quasi-coherence principle with respect to quotient algebras $\mathbb{I}/(i=1)$ in \zcref{prop:Idominance}; the co-dominance principle is dually identified with quasi-coherence with respect to quotient algebras $\mathbb{I}/(i=0)$ in \zcref{cor:dualisdominance}. This leads to an explicit computation of the partial map classifiers of spatial $\I$-algebras (\zcref{prop:liftingofalgebra}), and of spectra (\zcref{prop:liftofsober}).
    \item We give an explicit computation of the \emph{observational preorder} (\zcref{defn:specialisation}) for sober spaces (\zcref{lem:specorderofsober}) and spatial $\I$-algebras (\zcref{cor:specordonalgiscan}).
    \item A generalisation of Phoa's principle that applies to an arbitrary $\I$-algebra $A$ is stated (\zcref{def:gen-phoa}), and is shown to coincide with quasi-coherence with respect to polynomial algebras (\zcref{obs:sqcp-phoa-gen}). As a corollary, we see that spatiality of finitely generated free $\I$-algebras is equivalent to Phoa's principle (\zcref{cor:phoa-vs-quasicoherence}).
    \item In the case of $\sigma$-frames, quasi-coherence with respect to countably presented algebras implies chain completeness of $\I$ (\zcref{thm:complete}) assuming $\I$ is non-trivial. 
    This is the only place where it is important to work with $\sigma$-frames and not bounded distributive lattices (cf.~\zcref{rem:whynotdis}). This gives a full account of the axioms for synthetic domain theory as specified by \citet{hyland1990first}. 
  \end{enumerate}

  \item \emph{Local properties of spectra.} Though in general we do not assume the generic interval $\I$ to satisfy any local property such as linearity, we do discuss many instances of them and investigate their consequences.
  \begin{enumerate}

    \item \zcref{sec:locality} introduces various local conditions for $\I$, all of which will be compatible with quasi-coherence. This showcases the flexibility of this framework to encompass different flavours of domain theory.

    \item  More interestingly, without assuming the local properties globally, one can still show that $\I$ will be \emph{right orthogonal} to the maps that classifies these local properties (\zcref{specisnontrivial,specistriangulated,specis1t}). In general, any limiting diagram of spatial algebras will induce a localisation class containing $\I$; cf.\ \zcref{rem:limofalgloc}. This exhibits a new type of techniques in reasoning about domains in this framework.
    
    \item As a special case of locality, we also connect to the recent approach of synthetic (higher) category theory~\citep{riehl2017type,buchholtz2021synthetic,gratzer2024directed}. In particular, we show spectra will be synthetic categories (\zcref{thm:soberposet}). As another example, we also show $\omega$ is a synthetic category, in fact it satisfies \emph{all} the orthogonality conditions discussed in this paper except chain completeness (\zcref{thm:omegaortho}).
  \end{enumerate}

  \item \emph{New models for synthetic domain theory.} Consequently, we uncover a large family of new models for synthetic domain theory based on sites induced by countably presented $\sigma$-frames. We will discuss these models in \zcref{sec:model}, and compare them to the existing sheaf models for synthetic domain theory~\citep{FIORE1997151} in \zcref{subsec:compare}.
\end{enumerate}

\subsection{Style and notation of the paper}

Our motivation for domain theory encourages us to work in a context as general as possible, so as to allow future developments that connect more specific flavours of domain theory existing in the literature. This means that besides \zcref{sec:model} where we discuss models, throughout the paper we will work constructively in a type theory enriched with a generic model $\I$. The base type system we work with is intensional type theory with a universe satisfying function extensionality, which can be interpreted in any ($\infty$-)topos. 

We do not assume any additional assumption globally. Whenever we introduce an assumption under the \textbf{Axiom} environment, it should be viewed as introducing the \emph{content} of that assumption, rather than assuming it directly afterwards. In particular, the development is completely modular, and any additional assumption will be explicitly mentioned in each result.

In fact, we even take a step further. For the first half of this paper we will not work with bounded distributive lattices specifically. Instead, we assume to work with an arbitrary \emph{propositionally stable} theory in the sense of \zcref{defn:propositional}. In particular, the construction of the dominance and the computation of partial map classifier mentioned in \zcref{sec:contributions} (\ref{contribution:dominance}) works in this generality. Only from \zcref{sec:locality} onward will we then work more specifically with theories based on bounded distributive lattices.

\subsection*{Notation and preliminaries}
We work informally in the vernacular of univalent foundations~\citep{hottbook}.
By \emph{proposition} and \emph{set}, we will always use them in the sense of the HoTT Book~\citep{hottbook}, i.e.\ $-1$-types and $0$-types. The subuniverse of propositions and sets will be denoted as $\pp$ and $\stt$, respectively. Propositions are used to define subtypes $P \colon X \to \pp$, and we also write the dependent sum $\sum_{x:X}P(x)$ suggestively as $\scomp{x:X}{P(x)}$. In this case, for $x:X$ we also write $x\in P$ to denote the proposition $P(x)$. Notice that function extensionality implies $\pp$ and $\stt$ are both closed under dependent product, and $\stt$ furthermore is closed under dependent sums. For emphasis, we will also use $\fa xXP(x)$ to denote the dependent product $\prod_{x:X}P(x)$ when $P(x)$ is a family of propositions; in this case, we shall write $\ex xXP(x)$ for the propositional truncation $\pss{\sum_{x:X}P(x)}$. By \emph{existence} we will always mean the latter truncated version, and we may emphasise this by referring to \emph{mere existence}.\footnote{Here we deviate from the conventions of the HoTT Book~\citep{hottbook}.} Similarly, we use $P \vee Q$ to mean the truncated proposition $\pss{P + Q}$. We will write $\N$ for the type of natural numbers. For $n:\N$, we will also abuse notation and treat $n$ as the finite type of $n$ elements. 

Given an $\I$-algebra $A$, we will abuse notation by identifying elements of $\I$ as elements of $A$ via the structural map $\I \to A$. For an algebra $A$ and two lists of terms $a,b \colon J \to A$, 
we will write $A/a = b$ to indicate the \emph{quotient algebra} identifying $a_j$ with $b_j$ for all $j : J$. We also write $A[J]$ for the free algebra over $A$ generated from $J$, or we also explicitly mention the generators $A[\ms{j}_1,\cdots,\ms{j}_n]$. We write the \emph{coproduct} of $\I$-algebras as $A \otimes B$.

\section{Quasi-coherence and sober spaces}\label{sec:basics}

We will work in intensional type theory extended by a model $\I$ of some Horn theory $\T$. In this section we will first describe the quasi-coherence property, and briefly recap some of its elementary consequences. The results in this section applies to any Horn theory $\T$, thus at this stage we do not assume $\T$ to satisfy any additional properties.

As mentioned in \zcref{subsec:qc}, an $\I$-algebra is a $\T$-model $A$ equipped with a homomorphism $\I \to A$. Here by a $\T$-model we always mean a \emph{set} (in the sense of univalent foundations) equipped with a family of operations satisfying the axioms specified by $\T$. $\I\alg$ will denote the category of $\I$-algebras. We have a construction of spectra as a contravariant functor:

\begin{definition}[Spectrum]
  The \emph{spectrum} of an $\I$-algebra $A$ is defined to be the set of $\I$-algebra homomorphisms from $A$ into $\I$:
  \[ \spec A \coloneq \I\alg(A,\I)\text{.} \]
  This gives us a contravariant functor 
  \[ \spec \colon \I\alg\op \to \Set\text{,} \]
  acting on morphisms by pre-composition.
\end{definition}

\begin{definition}[Observation algebra]
  Given a set $X$, its \emph{algebra of observations} is the product of following $\I$-algebras:
  \[ 
    \opens X \coloneq \I^X = \prod_{x:X}\I
  \] 
  This gives us a functor \[ \opens \colon \Set \to \I\alg\op \]
  taking any set to its algebra of observations.
\end{definition}

\begin{proposition}\label{specrightadj}
  The spectrum construction is right adjoint to the functor taking a set $X$ to its algebra of observations:
  \[\begin{tikzcd}
    {\I\alg\op} & \Set
    \arrow[""{name=0, anchor=center, inner sep=0}, "\spec"', curve={height=18pt}, from=1-1, to=1-2]
    \arrow[""{name=1, anchor=center, inner sep=0}, "\opens"', curve={height=18pt}, from=1-2, to=1-1]
    \arrow["\dashv"{anchor=center, rotate=-90}, draw=none, from=1, to=0]
  \end{tikzcd}\]
\end{proposition}

\begin{proof}
  For any $\I$-algebra $A$ and set $X$, we have a natural equivalence 
  \[ (\spec A)^X \cong \I\alg(A,\I)^X \cong \prod_{x:X}\I\alg(A,\I)\cong \I\alg(A,\opens{X})\text{.} \qedhere\]
\end{proof}

\begin{convention}[Unit and counit]
  Following \citet{Taylor2011}, we shall write $\eta\colon 1_{\Set}\to {\spec}\circ \opens$ and $\iota\colon \opens\circ\spec\to 1_{\I\alg\op}$ for the unit and counit of the adjunction $\opens\dashv \spec\colon \I\alg\op\to \Set$ respectively. Due to the contravariance of the adjunction, the components of the counit are written $\iota_A\colon A\to \opens\spec{A}$ in the category of $\I$-algebras. Both $\eta_X\colon X\to \spec\opens{X}$ and $\iota_A\colon A\to \opens\spec{A}$ are the evident evaluation maps.
\end{convention}

Many sets can be naturally viewed as the spectrum of a certain $\I$-algebra. For example, we have $\spec \I \cong \I\alg(\I,\I) \cong 1$ because  $\I$ is the initial $\I$-algebra. Similarly, $\I$ itself it the spectrum of the polynomial $\I$-algebra $\I[\ms{i}]$, as we have $\spec \I[\ms{i}] \cong \I\alg(\I[\ms{i}],\I) \cong \I$. More generally:

\begin{example}[Cubes]\label{exm:cubesober}
  For any set $X$, let $\I[X]$ be the free $\I$-algebra generated by $X$. Then the set $\I^X$ is a spectrum, because we have
  \[ \spec \I[X] \cong \I\alg(\I[X],\I) \cong \I^X\text{.} \]
\end{example}

Notice $\I[X]$ is the $X$-fold coproduct of the polynomial $\I$-algebra $\I[\ms{i}]$, and its spectrum is the $X$-fold product of $\spec\I[\ms{i}] \cong \I$. In fact, all the spectra can be obtained as equalisers of powers of $\I$:

\begin{proposition}\label{prop:spectra-are-powers-of-the-interval}
  Every spectrum is an equaliser of powers of $\I$. 
\end{proposition}

To verify \zcref{prop:spectra-are-powers-of-the-interval}, we recall some results concerning adjunctions of descent type. This notion was first given in \citet{BarrMichael1985Ttat}, and the following characterisation appeared in \citet{kelly1993adjunctions}

\begin{definition}
  An adjunction $F\dashv U\colon \mathcal{A}\to \mathcal{C}$ is said to be of \emph{descent type} when either of the equivalent conditions hold:
  \begin{enumerate}
    \item The comparison functor $\Phi_{UF}\colon \mc{A} \to \mc{C}^{UF}$ into the Eilenberg--Moore category of the monad $UF$ on $\mathcal{C}$ is fully faithful.
    \item For each $A\in \mc{A}$, the counit $\epsilon_A\colon FUA\to A$ is a coequaliser in $\mc{A}$:
    \[\begin{tikzcd}
      FUFUA \ar[r, shift left, "\epsilon_{FUA}"] \ar[r, shift right, "{FU\epsilon_A}"'] & FUA \ar[r, two heads, "\epsilon_{A}"] & A
    \end{tikzcd}\]  
  \end{enumerate}
\end{definition}

\begin{proposition}\label{lem:horn-free-models-descent}
  For a Horn theory $\T$, the adjunction $F\dashv U\colon \mmod\T\to\Set$ is of descent type.
\end{proposition}

\begin{proof}
  Consider a presentation $\T = (\Sigma_\T,H_\T)$ where $\Sigma_\T$ is a signature and $H_\T$ is a set of Horn clauses over this signature; we may of course regard $\Sigma_\T$ as an algebraic theory without any equations. We may then decompose the adjunction $F\dashv U\colon \mmod\T\to \Set$  into the composite adjunction:
  \[\begin{tikzcd}
    {\mmod\T} & {\mmod{\Sigma_\T}} & \Set
    \arrow[""{name=0, anchor=center, inner sep=0}, "{U_1}", curve={height=-15pt}, hook', from=1-1, to=1-2]
    \arrow[""{name=1, anchor=center, inner sep=0}, "{F_1}", curve={height=-15pt}, from=1-2, to=1-1]
    \arrow[""{name=2, anchor=center, inner sep=0}, "{U_0}", curve={height=-15pt}, from=1-2, to=1-3]
    \arrow[""{name=3, anchor=center, inner sep=0}, "{F_0}", curve={height=-15pt}, from=1-3, to=1-2]
    \arrow["\dashv"{anchor=center, rotate=90}, draw=none, from=1, to=0]
    \arrow["\dashv"{anchor=center, rotate=90}, draw=none, from=3, to=2]
  \end{tikzcd}\]

  The existence and reflectivity of $F_1\dashv U_1\colon \mmod\T\to\mmod{\Sigma_\T}$ is established by Remark~3.19~(1) of \citet{adamek1994locally}, as $\mmod\T$ is a ``quasivariety'' in the sense of \emph{op.\ cit.} The right-hand adjunction $F_0\dashv U_0\colon \mmod{\Sigma_T}\to \Set$ is monadic (i.e.\ of \emph{effective} descent type) because $\Sigma_\T$ is algebraic. Corollary 3.3 of \citet{kelly1993adjunctions} implies that the composition of a reflective adjunction followed by an adjunction of descent type is of descent type. 
\end{proof}

We are now prepared to verify \zcref{prop:spectra-are-powers-of-the-interval}.

\begin{proof}[Proof of \zcref{prop:spectra-are-powers-of-the-interval}]
  We wish to show that every spectrum $X=\spec A$ an equaliser of powers of the interval.
  To that end, we note that the adjunction $F\dashv U\colon \I\alg \to\Set$ is of descent type by \zcref{lem:horn-free-models-descent}, since the theory of $\I$-algebras is Horn. Hence, the following is a coequaliser in $\I\alg$: 
  \[
  \begin{tikzcd}
    \I[\I[A]] \ar[r, shift left, "\epsilon"] \ar[r, shift right, "{\I[\epsilon]}"'] & \I[A] \ar[r, two heads, "\epsilon"] & A
  \end{tikzcd}
  \]

  Because $\spec\colon \I\alg\op\to\Set$ is a right adjoint (\zcref{specrightadj}), it sends colimits of algebras to limits of sets; hence, the following is an equaliser:
  \[
  \begin{tikzcd}
    \spec A \ar[r, hook] & \spec\I[A] \cong \I^A \ar[r, shift left] \ar[r, shift right] & \spec\I[\I[A]]\cong\I^{\I[A]}
  \end{tikzcd}
  \qedhere
  \]
\end{proof}

\begin{remark}[Spectra are replete]\label{rem:specarereplete}
  Recall the notion of \emph{replete} type given by \citet{hyland1990first}: $X$ is said to replete when it is right orthogonal to any map that $\I$ is right orthogonal to. Throughout the paper, by right orthogonality we always mean the \emph{internal} notion: $X$ is right orthogonal to a map $f$ iff the restriction map $X^f$ is an equivalence. Thus, an object is replete iff it belongs to the smallest internal localisation class containing $\I$. In particular, replete objects form an exponential ideal. This implies a spectrum $\spec A$ as an equaliser of powers of $\I$ will also be replete.
\end{remark}

We will now define spatial algebras and sober spaces as the \emph{fixed points} of the adjunction in \zcref{specrightadj}:

\begin{definition}[Spatial algebra]
  An $\I$-algebra $A$ is said to be \emph{spatial} when it is a fixed point of the adjunction $\opens \dashv \spec$, i.e.\ the component $\iota_A\colon A \to \opens \spec A$ of the counit is an equivalence of $\I$-algebras.
\end{definition}

We emphasise that the equivalence $A \cong \opens\spec A$ is not only an equivalence of types, but an equivalence of $\I$-algebras. This implies that the algebraic structure on $A$ can be viewed as pointwise induced by that on $\I$. 

\begin{remark}[Spatial algebras are replete]\label{rem:qcreplete}
  The underlying set of any spatial algebra $A$ is also replete, as we have $A\cong \opens\spec{A}$ and hence the underlying set of $A$ is $\I^{\spec{A}}$.
\end{remark}

\begin{example}\label{exm:intervalqc}
  $\I$ itself by definition is spatial. We have seen that $\spec \I \cong 1$. Under this equivalence, the canonical map 
  \[ \I \to \opens\spec \I \cong \I \]
  is exactly the identity on $\I$.
\end{example}

\begin{definition}[Sober sets]
  We say a set $X$ is \emph{sober} if it is a fixed point of the adjunction $\opens \dashv \spec$, i.e.\ the unit map $\eta_X\colon X \to \spec\opens X$ is an equivalence.
\end{definition}

In the future whenever we indicate $X$ is sober with $X \cong \spec A$, we will assume $A$ is the spatial $\I$-algebra $A \cong \opens X$.

\begin{remark}[Spatiality and sobriety as general properties]
  Our focus on spatiality and sobriety indicates an important difference between our approach and the ones taken in the related works~\citep{Cherubini_Coquand_Hutzler_2024,cherubini2024foundation}. There, the notion is directly connected to the size of presentation, i.e.\ they are spectra of finitely presented or countably presented algebras. Since we aim for a completely modular development, our notion is more general, which a priori does not favour a class of size-related algebras. 
\end{remark}

Being the fixed points of an adjunction, spatial $\I$-algebras and sober spaces induces certain duality results. The first half of the following result appears in Prop.~2.2.1 of \citet{Cherubini_Coquand_Hutzler_2024} (for the more restrictive definition of quasi-coherence). We include the result here for completeness since our assumption is slightly more general.

\begin{proposition}\label{prop:duality}
  Let $A$ be spatial. For any $\I$-algebra $B$, the canonical map $\I\alg(B,A) \to (\spec B)^{\spec A}$ is an equivalence. Similarly, if $X$ is sober, for any set $Y$ the canonical map $X^Y \to \I\alg(\opens X,\opens Y)$ is an equivalence.
\end{proposition}
\begin{proof}
  This holds generally for fixed points of an adjunction.
\end{proof}

\section{Open propositions and the dominance property}\label{sec:dominance}

Notice that up till this point we have not used any special property of $\T$ rather than the fact that it is a Horn theory. However, to move closer to the intended application in domain theory, we start by assuming our theory $\T$ is \emph{propositionally stable} in the following sense: 

\begin{definition}[Propositionally stable theory]\label{defn:propositional}
  We say a Horn theory $\T$ is \emph{propositionally stable}, if it extends the theory of bounded meet-semi-lattices (1, $\wedge$), and truth of an element is computed by slicing: For any $\T$-model $A$ and element $a:A$, the quotient $A/a=1$ is given by the restriction map
  \[ {a \wedge -} \colon A \surj A/a\text{,} \]
  where the slice $A/a$ by definition is ${\cv} a \coloneq \scomp{b:A}{b\le a}$, with $\le$ the partial order on $A$ induced by the bounded meet-semi-lattice structure.
\end{definition}

\begin{remark}[Examples of propositionally stable theories]
  $\mbb M$ of bounded meet-semi-lattices, $\mbb D$ of bounded distributive lattices, $\mbb H$ of Heyting algebras, and $\mbb B$ of Boolean algebras are all examples of propositionally stable theories. In fact, all finitary quotients of $\mbb H$ or $\mbb B$ will again be propositionally stable. More generally, for any propositionally stable theory $\T$ and any $\T$-model $D$, the theory of $D$-algebras will again be so, because quotients of $D$-algebras are computed the same as quotients of models of $\T$.
\end{remark}

\begin{remark}[The theory of $\sigma$-frames]\label{rem:sigmaframe}
  A $\sigma$-frame $A$ is a lattice with finite meets and countable joins satisfying the distributivity axiom 
  \[ a \wedge \bigvee_{n:\N} b_n = \bigvee_{n:\N} a \wedge b_n\text{,} \]
  for any $a$ in $A$ and $b \colon \N \to A$. We can axiomatise $\sigma$-frames as an \emph{infinitary} algebraic theory $\mbb S$, extending the theory of bounded meet-semi-lattices with a constant $0$ and a function symbol $\bigvee$ with \emph{countable} arity. (Such an axiomatisation can be found e.g.\ in \citet[Exa.~3.26]{adamek1994locally}.) It is easy to see that $\mbb S$ is propositionally stable. 

  A notable fact about $\sigma$-frames is that their finitary behaviours are exactly the same as bounded distributive lattices, in the sense that for any $\sigma$-frame $A$, the finitely presented $\sigma$-frame $A[n]/R$ coincides with the finitely presented \emph{bounded distributive lattice} $A[n]/R$. Thus, there will be no difference between working with $\sigma$-frames or bounded distributive lattice when we work with finitely presented $\I$-algebras. The salience of countable joins will emerge only in \zcref{sec:infdomain} when we prove chain completeness of $\I$.
\end{remark}

\begin{remark}[Propositionally op-stable theories]\label{rem:opprop}
  There is an evident ``opposite'' notion of being propositionally stable, which requires $\T$ to extend the theory of bounded join-semi-lattices and the quotient $A/a = 0$ for any $a:A$ is computed by ${a \vee -} \colon A \surj a/A$. Every propositionally stable theory corresponds to a propositionally op-stable theory by dualising: The theories of join-semi-lattices and coHeyting algebras are propositionally op-stable; Boolean algebras and distributive lattices dualise to themselves; the theory of $\sigma$-frames is also propositionally op-stable.
\end{remark}

For a propositionally stable theory $\T$, we think of the generic model $\I$ as certain interval object, since it is equipped with a partial order. In this case, more types can be realised as spectra. The important examples are \emph{simplices}:

\begin{example}\label{exm:simplicessober}
  For any $n : \N$, let $\Delta^n \hook \I^n$ be the $n$-simplex
  \[ \Delta^n \coloneq \scomp{i \colon n \to \I}{i_1 \ge \cdots \ge i_n}\text{.} \]
  This type is indeed a spectrum, since by definition we have
  \[ \Delta^n \cong \spec\I[\ms{i}_1,\cdots,\ms{i}_n]/\ms{i}_1\ge\cdots\ge \ms{i}_n\text{.} \]
  To simplify later discussion, we might also introduce the following types isomorphic to simplices above,
  \[ \Delta_n \coloneq \scomp{i \colon n \to \I}{i_1 \le \cdots \le i_n}\text{.} \]
\end{example}

The constant $1 : \I$, which is the top element in $\I$, induces a predicate
\[ \istsym\colon \I \to \pp\text{,} \]
which takes $i : \I$ to the proposition $i = 1$. Using $\istsym$, we can think of the spectrum $\spec A$ of an $\I$-algebra $A$ as its set of \emph{models}, with a satisfaction relation $\models$ as follows:

\begin{definition}[Satisfaction]
  For any $\I$-algebra $A$ and elements $a:A$ and $x:\spec A$, we say $x$ \emph{satisfies} $a$, written $x \models a$, if $\ist{xa}$.
\end{definition}

For instance, the satisfaction relation can be used to carve out sober subspaces consisting of models satisfying some element in $A$:

\begin{example}
  For any $\I$-algebra $A$ and $a:A$, the subtype $D_a\subseteq \spec A$ defined below is again a spectrum:
  \[ D_a \coloneq \scomp{x : X}{x \models a} \cong \spec A/a\text{.} \]
\end{example}

The first observation is that $\istsym$ takes $i : \I$ to a spectrum:

\begin{lemma}\label{lem:openpropsober}
  The spectrum $\spec\I/i$ for any $i:\I$ is equivalent to the proposition $\ist{i}$.
\end{lemma}
\begin{proof}
  We have $\spec \I/i \cong \I\alg(\I/i,\I)$ by definition. Since $\I$ is the initial $\I$-algebra and $\I/i$ is a quotient of $\I$, any homomorphism $\I/i \to \I$ will be unique, hence $\spec\I/i$ is a proposition. It thus suffices to show $\spec\I/i$ implies $\ist{i}$. Because $\T$ is propositionally stable, an $\I$-algebra morphism $f \colon \I/i \to \I$ is a commutative diagram as follows,
  \[
  \begin{tikzcd}
    \I/i \ar[rr, dashed, "f"] & & \I \\ 
    & \I \ar[ul, two heads, "i \wedge -"] \ar[ur, equal]
  \end{tikzcd}
  \]
  In this case, we have
  \[ i = f(i \wedge i) = f(i) = 1\text{.} \]
  The final identity holds since $f$ preserves the top element. 
\end{proof}

The goal of this section is to show that, under a suitable quasi-coherence assumption, $\istsym\colon \I \to \pp$ makes $\I$ a \emph{dominance} in the sense of \citet{rosolini1986continuity}. In fact, at the end of this section we will provide an exact characterisation of the dominance property of $\I$ as a quasi-coherence principle. It turns out that it suffices to require all slices of $\I$ to be spatial:

\begin{definition}[Stable spatiality and sobriety]
  An $\I$-algebra $A$ is \emph{stably spatial}, if for any $i:\I$, the slice $A/i$ is spatial. A sober space $X \cong \spec A$ is \emph{stably sober} if $A$ is stably spatial.
\end{definition}

\PrintAxiomSQCI

As a first consequence, we show the following important lemma:

\begin{lemma}[\AxiomSQCI]\label{lem:intconserve}
  The interval $\I$ is \emph{conservative}: For any $i,j : \I$
  \[ (\ist{i} \eq \ist{j}) \eq i = j\text{.} \]
\end{lemma}
\begin{proof}
  If $i \le j$ then trivially $\ist{i} \to \ist{j}$. Thus it suffices to show
  \[ (\ist{i} \to \ist{j}) \to i \le j\text{.} \]
  Take $i,j$ with $\ist{i} \to \ist{j}$. By \zcref{lem:openpropsober} and (\AxiomSQCI), each $\ist{i}$ and $\ist{j}$ will be sober. Then we get a restriction map between $\I$-algebras,
  \[ \I/j \cong \I^{\ist{j}} \to \I^{\ist{i}} \cong \I/i\text{.} \]
  Similar to the proof of \zcref{lem:openpropsober}, the existence of such an $\I$-algebra homomorphism implies $i \le j$. 
\end{proof}

Thus, under (\AxiomSQCI), we may view the generic algebra $\I$ as a subuniverse of \emph{open} propositions via the embedding $\istsym$:

\begin{definition}[Open proposition]
  We say that a proposition $p$ is \emph{open} when $p$ is equivalent to $\ist{i}$ for some $i$ in $\I$.
\end{definition}

Under (\AxiomSQCI), by \zcref{lem:intconserve} the $i$ that $p \eq \ist{i}$ will be \emph{unique}, thus being open is itself a proposition. Also, open propositions are evidently closed under finite conjunctions, since $\istsym$ preserves them. 

More generally, we may define the notion of open subtypes:

\begin{definition}[Open subtype]
  A subtype $U$ of $X$ is \emph{open} if for any $x:X$ the proposition $x\in U$ is open.
\end{definition}

If $X$ is sober, open subtypes are easy to classify:

\begin{lemma}[\AxiomSQCI]\label{lem:openofsobergivesalgebra}
  Let $X \cong \spec A$ be sober. Then any open subset of $X$ is of the form $D_a$ for some $a:A$.
\end{lemma}
\begin{proof}
  Let $U$ be an open subtype of $X$, so that for any $x:X$ there exists some (by conservativity, necessarily unique) element $i:\I$ that $x \in U \eq \ist{i}$.
  Hence we have a map $\qsi a \colon X \to \I$ with $x \in U \eq \ist{\qsi ax}$. Because $X$ is sober, we may identify $\qsi a$ with an element $a:A$, where $\qsi ax = xa$. This way,
  \[ \fa xX x \in U \eq x \models a\text{,} \]
  and thus $U$ is the subtype $D_a$.
\end{proof}

\begin{proposition}[\AxiomSQCI]\label{prop:Idominance}
  The type of opens $\I$ forms a dominance.
\end{proposition}
\begin{proof}
  Since $\istsym\colon \I \hook \pp$ by definition is closed under finite meets, it suffices to show that an open subtype of an open proposition is again an open proposition. Suppose $p$ is an open proposition and $q$ is an open subtype of $p$. By definition $p \eq \ist{i}$ for some $i:\I$. Since $\ist{i}$ is sober, \zcref{lem:openofsobergivesalgebra} implies that $q$ is $D_j$ for some $j : \I/i$. Since $\T$ is propositionally stable, equivalently $j$ can be viewed as an element $j : \I$ with $j \le i$. Hence,
  \[ q \eq D_j \eq \spec(\I/i)/j \eq \spec\I/j\text{,} \]
  which implies $q$ is also an open proposition.
\end{proof}

\begin{remark}[Dominance without choice]\label{rem:dominancewithoutchoice}
  The related work of \citet{Cherubini_Coquand_Hutzler_2024} contains a similar construction of a dominance from a generic \emph{ring}. However, the proof there relies on certain classical axiom, which they call \emph{Zariski local choice}. The reason that the construction here does not require a similar assumption is precisely due to conservativity of $\I$: the dominance is $\I$ itself, rather than an image of $\I$. In this way, we can transform an open subtype of $X$ to a map $X \to \I$, which by quasi-coherence can be further identified as an element of the algebra if $X$ is sober. 
\end{remark}

We end this section by showing that $\I$ being stably spatial, i.e.\ all slices of $\I$ are spatial, is a \emph{characterisation} the dominance property of $\I$. Under this assumption, we can prove the (surprising) fact that all spatial algebras will indeed be stably spatial:

\begin{theorem}\label{thm:chardominance}
  The following are equivalent:
  \begin{enumerate}
    \item (\AxiomSQCI) holds, i.e.\ $\I$ is stably spatial;
    \item $\I$ form a dominance;
    \item All spatial $\I$-algebras are stably spatial.
  \end{enumerate}
\end{theorem}
\begin{proof}
  We have shown (1) $\nt$ (2) in~\zcref{prop:Idominance}. Evidently, (3) $\nt$ (1) since by~\zcref{exm:intervalqc} $\I$ is spatial. We first show (2) $\nt$ (1). Suppose $\I$ is a dominance and take any $i:\I$. We construct an explicit inverse $s \colon \I^{\ist{i}} \to \I/i$ by defining $s$ as follows,
  \[ s(j) \coloneq \sum_{p:\ist i} j(p). \]
  By construction $s(j) \le i$, since if $(p,x) : s(j)$, then $p : \ist i$ which implies $i = 1$, hence $x \le i$. To show this is an inverse of $\eta \colon \I/i \to \I^{\ist i}$, for any $j \le i$,
  \[ s\eta(j) \eq \sum_{p:\ist i}j \eq j. \]
  On the other hand, for any $j:\I^{\ist i}$, to show $j = \eta s(j)$ by function extensionality it suffices to assume $i = 1$, viz.\ there is $p:\ist i$. In this case, 
  \[ \eta s(j) = \eta j(p) = \ld{q}{\ist i} j(p). \]
  But in this case, for any $q:\ist i$, $q = p$, thus $j(q) = j(p)$, which by function extensionality again it follows that 
  \[ \eta s(j) = j. \]
  This shows $s$ is indeed an inverse of $\eta$, hence $\I/i$ is spatial.

  Now we show (1) $\nt$ (3). Suppose $A$ is a spatial $\I$-algebra. Since the slice $A/i$ is the quotient algebra $A/i = 1$, it follows that $A/i \cong A \otimes \I/i$. By~\zcref{specrightadj}, we have
  \[ \spec(A/i) \cong \spec A \times \spec{\I/i} \cong \spec A \times \ist i. \]
  Since $\I/i$ is spatial, the following canonical isomorphisms hold,
  \[ \opens{\spec(A/i)} \cong (\I/i)^{\spec A} \cong \scomp{U:\opens{\spec A}}{\fa x{\spec A} U(x) \le i}. \]
  Since $A$ is spatial, $A \cong \opens{\spec A}$. In fact, spatiality also implies that there are ``enough points'' in the spectrum to detect its order (cf.~\zcref{lem:canon-beh-coincide}),
  \[ \opens{\spec(A/i)} \cong \scomp{a:A}{\fa x{\spec A} x(a) \le i} \cong \scomp{a:A}{a \le i} \cong A/i. \]
  This shows $A/i$ is also spatial.
\end{proof}

\section{Partial map classifier}\label{sec:lifting}
Given the interval $\I$, we can construct internally the partial map classifier. For any type $X$, it is given by
\[ L X \coloneq \sum_{i:\I}X^{\ist{i}}\text{.} \]
The functoriality is easy to express: For any $f \colon X \to Y$, we have
\[ L f(i,x) \coloneq (i,\ld{w}{\ist{i}}fxw)\text{.} \]
There is an evident unit $\eta_L \colon X \to L(X)$, where
\[ \eta_L(x) \coloneq (1,\ld\_ 1 x)\text{.} \]
The unit $\eta_L$ then classifies the partial maps out of $X$ with an \emph{open} domain. If $\I$ forms a dominance, then there is also a multiplication $\mu \colon L(L X) \to L X$, where $\mu$ takes any $(i,u)$ with $i : \I$ and $u \colon \ist{i} \to L X$ first to $(j,x)$, where $j$ is the dependent sum
\[ j \coloneq \sum_{w:\ist{i}} (uw)_0\text{,} \]
and $x \colon \ist{j} \to X$ is the partial element such that for $w : \ist{i}$ and $v : (uw)_0$
\[ x(w,v) \coloneq (uw)_1(v)\text{.} \]

\begin{example}
  By definition, it is easy to see that
  \[ L 1 \cong \sum_{i:\I}\ist{i} \cong \I\text{.} \]
\end{example}

For synthetic domain theory, the object of particular importance is the partial map classifier of the interval $\I$ itself. The following computation works more generally for all stably spatial algebras:

\begin{proposition}\label{prop:liftingofalgebra}
  If $A$ is stably spatial, then we have
  \[ L A \cong \sum_{i:\I}A/i\text{.} \]
\end{proposition}
\begin{proof}
  By assumption for $i : \I$, the quotient $A/i$ is again spatial. Now notice that we do have
  \[ A/i \cong A \otimes \I/i\text{.} \]
  By \zcref{specrightadj} and spatiality of $A/i$ and $A$,
  \[ A/i \cong \I^{\spec(A \otimes \I/i)} \cong \I^{\spec A \times \spec\I/i} \cong A^{\spec \I/i} \cong A^{\ist{i}}\text{.} \]
  This way, it follows that 
  \[ L A \cong \sum_{i:\I}A^{\ist{i}} \cong \sum_{i:\I}A/i. \qedhere \]
\end{proof}

In other words, an element $(i,a) : L A$ is simply a pair $i : \I$ and $a : A$ that $a \le i$. This way, we can easily compute $L \I$:

\begin{corollary}[\AxiomSQCI]
  $L \I \cong \Delta^2$.
\end{corollary}
\begin{proof}
  By \zcref{prop:liftingofalgebra} and the assumption (\AxiomSQCI),
  \[ L\I \cong \sum_{i:\I}\I/i \cong \scomp{i,j : \I}{i \ge j} \cong \Delta^2\text{.} \]
  The second equivalence is again due to $\T$ being propositionally stable.
\end{proof}

\begin{remark}[Algebra vs.\ geometry]\label{rem:alggeoI}
  One interesting fact to note here is that, though $L\I$ by computation is a dependent sum of algebras, it is naturally equivalent to a \emph{spectrum}. In some sense the source is the assumption that $\T$ is propositionally stable, which allows us to identify the algebraic object $\I/i$ as a subset $\scomp{j : \I}{j \le i}$. There will be more examples of such nature once we work more specifically with bounded distributive lattices; cf.\ \zcref{prop:simplicesasalgebra}.
\end{remark}

More generally, for domain theoretic applications we would want to compute the partial map classifier of $\Delta^2$, or $\Delta^n$ for any $n:\N$. For this purpose, we observe that we can more generally compute them for any spectra (not only sober spaces):

\begin{proposition}[\AxiomSQCI]\label{prop:liftofsober}
  For an $\I$-algebra $A$, we have
  \[ L\spec A \cong \sum_{i:\I}\I\alg(A,\I/i)\text{.} \]
\end{proposition}
\begin{proof}
  By (\AxiomSQCI) and \zcref{lem:openpropsober}, $\ist{i}$ is sober for any $i:\I$. By the duality result in \zcref{prop:duality}, we have
  \[ L\spec A \cong \sum_{i:\I}(\spec A)^{\ist{i}}\cong \sum_{i:\I}\I\alg(A,\I/i). \qedhere \]
\end{proof}





\begin{corollary}[\AxiomSQCI]
  For any $n : \N$, we have
  \[ L\Delta^n \cong \Delta^{n+1}\text{,} \]
  and the unit $\eta_L \colon \Delta^n \hook L\Delta^n \cong \Delta^{n+1}$ takes $i_1 \ge \cdots \ge i_n$ in $\Delta^n$ to $1 \ge i_1 \ge \cdots \ge i_n$ in $\Delta^{n+1}$. 
\end{corollary}
\begin{proof}
  By \zcref{prop:liftofsober}, for the spectrum $\Delta^n \cong \spec\I[n]/\ms{i}_1 \ge \cdots \ge \ms{i}_n$,
  \begin{align*}
    L\Delta^n
    &\cong \sum_{i:\I}\I\alg(\I[n]/\ms{i}_1\ge\cdots\ge \ms{i}_n,\I/i) \\
    &\cong \scomp{i_1,\cdots,i_n:\I/i}{i_1 \ge \cdots \ge i_n} \\
    &\cong \scomp{i,i_1,\cdots,i_n:\I}{i \ge i_1 \ge \cdots \ge i_n} \\
    &\cong \Delta^{n+1}
  \end{align*}
  Again the third steps holds since $\T$ is propositionally stable.
\end{proof}

\section{Distributive lattices and locality}\label{sec:locality}

One important finitary axiom for synthetic domain theory is \emph{Phoa's principle}, which we have mentioned in~\zcref{subsec:classtopphoa} is a consequence of the quasi-coherence principle for bounded distributive lattices. 
Thus, we now work with the theory $\mbb D$ of bounded distributive lattices, or more generally the theory of $D$-algebra for some bounded distributive lattice $D$. As indicated in \zcref{rem:sigmaframe}, computation for finitely presented algebras 
will not change if we replace bounded distributive lattices by $\sigma$-frames. Hence, the results here and in \zcref{sec:order-theoretic-structure} will be applicable to $\sigma$-frames as well, where only the properties of finitely presented $\I$-algebras matter.

As mentioned in \zcref{rem:opprop}, the theory of bounded distributive lattices is also \emph{propositionally op-stable}, and similarly for $\sigma$-frames. Hence, dual versions of the previous results will also hold by symmetry, when we exchange $1$ for $0$ and $\wedge$ for $\vee$. Thus, we first record some simple corollaries for the dual structure.

\subsection{The dual dominance and co-partial map classifier}
In this case, $\I$ has a minimal element $0$. Notice that the constant $0$ also induces a predicate on $\I$,
\[ \isfsym\colon \I \to \pp\text{,} \]
which takes any $i : \I$ to $i = 0$. Equivalently, we can view it as the $\ms t$ map for $\I\op$.

\begin{corollary}
  If $\I\op$ is stably spatial, then $\isfsym\colon \I \to \pp$ is an embedding.
\end{corollary}

We will now call propositions in the image of $\isfsym$ \emph{closed}, and a subtype classified by $\isfsym$ a \emph{closed subtype}. Analogously to the previous \zcref{prop:Idominance}, closed propositions also form a dominance if $\I\op$ is stably spatial:

\begin{corollary}\label{cor:dualisdominance}
  If $\I\op$ is stably spatial, then $\isfsym$ forms a dominance.
\end{corollary}

There is an accompanying ``co-partial map classifier'' $T$ for partial maps with \emph{closed} domains, as \citet{hyland1990first} pointed out. Concretely, for any type $X$,
\[ T X \coloneq \sum_{i:\I} X^{\isf{i}}\text{.} \]
And analogously to \zcref{prop:liftingofalgebra}, under the assumption that $\I\op$ is stably spatial, we can explicitly compute 
\[ T\Delta^n \cong \Delta^{n+1}\text{,} \]
where now the unit $\eta_T \colon \Delta^n \hook T\Delta^n \cong \Delta^{n+1}$ takes a sequence $i_1 \ge \cdots \ge i_n$ to $i_1 \ge \cdots \ge i_n \ge 0$ in $\Delta^{n+1}$. 

To maintain this symmetry, we also introduce the following axiom:

\PrintAxiomSQCID

The remaining part of this section will introduce various locality axioms for the interval $\I$, and discuss some of their consequences with the quasi-coherence axiom.

\subsection{Non-triviality}

As a start, a minimal requirement for synthetic domain theory to model divergent computation is for the dominance $\I$ to be closed under falsum. For this, we introduce the following minimal amount of non-triviality:

\PrintAxiomNT

Geometrically, the proposition $0 = 1$ can be identified with the spectrum $\spec\I/0=1$. Hence, (\AxiomNT) states that the unique map $\emp \hook \spec\I/0=1$ is an equivalence.
This way, $\emp \eq \ist{0}$ now will be sober, and it is both an open and closed proposition. As mentioned in introduction, (\AxiomNT) together with a strong enough quasi-coherence principle will imply all of Hyland's axioms for synthetic domain theory. Before we discuss that in \zcref{sec:infdomain}, a minimal amount of quasi-coherence provided by (\AxiomSQCI) already implies a lot of elementary properties. As a first consequence of (\AxiomNT) with (\AxiomSQCI), we can show that $\I$ is almost the Boolean algebra $2$ in the following sense:

\begin{proposition}[\AxiomNT, \AxiomSQCID]\label{prop:filed}
  For any $i : \I$ we have:
  \[ \neg \ist{i} \eq \isf{i}, \quad \neg\isf{i} \eq \ist{i}\text{.} \]
  In particular, the embedding $2 \hook \I$ induced by $0,1$ is $\neg\neg$-dense:
  \[ \fa i\I \dneg(\ist{i} \vee \isf{i})\text{.} \]
\end{proposition}
\begin{proof}
  If $\isf{i}$ then $\neg\ist{i}$ by (\AxiomNT). On the other hand, by the conservativity result in \zcref{lem:intconserve} and (\AxiomNT), if $i \neq 1$ then $i = 0$ since $0 \neq 1$. The dual case also follows by symmetry. Now given $i :\I$, if $\neg(\ist{i} \vee \isf{i})$, which implies $\neg\ist{i} \vee \neg\isf{i}$, equivalently $\isf{i} \wedge \ist{i}$, contradictory to (\AxiomNT). Hence, $\neg\neg(\ist{i} \wedge \isf{i})$.
\end{proof}

This allows us to observe that the open and closed propositions have a bijective correspondence between them:

\begin{corollary}[\AxiomNT, \AxiomSQCID]\label{cor:opendnegclose}
  For any proposition $p$, $p$ is open iff $\neg p$ is closed and vice versa. Furthermore, open and closed propositions are $\dneg$-stable.
\end{corollary}

On the other hand, \zcref{sec:order-theoretic-structure} will show that under the axiom (\AxiomNT), quasi-coherence implies the interval $\I$ will \emph{never} be isomorphic to $2$. 

\subsection{Locality}

A slightly stronger axiom which is common in the practice of domain theory is that the dominance $\I$ should furthermore be closed under all finite disjunctions. For this, we consider the following axiom:

\PrintAxiomL

Geometrically, the second condition states that the following two inclusions are jointly surjective. 
  \[ 
  \begin{tikzcd}
    \I \ar[r, hook, "{i \mapsto (i,1)}"] & \scomp{i,j : \I}{\ist{i\vee j}} & \I \ar[l, hook', "{j \mapsto (1,j)}"']
  \end{tikzcd}
  \]

Under (\AxiomL), the dominance $\I$ will be closed under finite disjunctions. Furthermore, it also implies more elementary properties about $\I$:

\begin{lemma}[\AxiomL, \AxiomSQCID]\label{lem:intisnotBoolean}
  $\I$ has no non-trivial complemented elements, i.e.\ if $i : \I$ is complemented, then $\ist{i} \vee \isf{i}$. 
\end{lemma}
\begin{proof}
  Suppose $i$ has a complement $j$. Then by (\AxiomL) we have
  \[ 1 \eq \ist{i\vee j} \eq \ist{i} \vee \ist{j}\text{,} \]
  and similarly,
  \[ \emp \eq \ist{i\wedge j} \eq \ist{i} \wedge \ist{j}\text{.} \]
  It follows that $\ist{j} \eq \neg \ist{i}$, and by \zcref{prop:filed}, $\ist{j} \eq \isf{i}$. Hence, $\ist{i} \vee \isf{i}$.
\end{proof}

\zcref{lem:intisnotBoolean} above allows us to realise 2 as a \emph{spectrum}:

\begin{example}\label{exm:2issober}
  Consider the algebra $B \coloneq \I[\ms{i},\ms{j}]/{\ms{i}\wedge \ms{j} =0,\ms{i}\vee \ms{j} = 1}$. By construction, $B$ classifies complemented elements in $\I$-algebras. Thus, assuming (\AxiomL) and (\AxiomSQCI), by \zcref{lem:intisnotBoolean}
  \[ \spec B \cong \scomp{i,j : \I}{i \wedge j = 0 \wedge i \vee j = 1} \cong 2\text{.} \]
  In particular, this means that for any sober space $X \cong \spec A$, $2^X$ will be equivalent to the set of complemented elements in $A$ by \zcref{prop:duality}.
\end{example}

As another example of a new spectrum (\AxiomL) allows us to define:

\begin{example}\label{exm:hornsober}
  Consider the right outer horn $\Lambda^2_2$, which we may describe by the following pushout:
  \[
    \begin{tikzcd}
      1 \ar[d, "1"'] \ar[r, "1"] & \I \ar[d] \\
      \I \ar[r] & \Lambda^2_2
      \arrow["\lrcorner"{anchor=center, pos=0.125, rotate=180}, draw=none, from=2-2, to=1-1]
    \end{tikzcd}
  \]
  By viewing $\Lambda^2_2$ as a subspace of $\I^2$, we may identify it as follows:
  \[ \Lambda^2_2 \cong \scomp{(i,j) : \I^2}{\ist{i} \vee \ist{j}}\text{.} \]
  Assuming (\AxiomL), $\ist{i} \vee \ist{j} \eq \ist{i\vee j}$, it follows that 
  \[ \Lambda^2_2 \cong \spec \I[\ms{i},\ms{j}]/\ms{i} \vee \ms{j}\text{.} \]
\end{example}

Axiom (\AxiomL) clearly has a dual counterpart:

\PrintAxiomCL

Similarly to \zcref{exm:hornsober}, this allows one to identify the left outer horn $\Lambda^2_0$ as a spectrum. Of course we can combine the axioms (\AxiomL) and (\AxiomCL). 

\subsection{Linearity and coskeletality}\label{sec:linearity-and-coskeletality}

An axiom even stronger than both (\AxiomL) and (\AxiomCL) is the simplicial axiom (\AxiomSL):

\PrintAxiomSL

Geometrically, the simplicial axiom states that the two simplices $\Delta^2,\Delta_2$ covers the square $\Delta^2 \cup \Delta_2 \cong \I^2$.
The strongest locality axiom states that the simplicial structure is truncated to level 1:

\PrintAxiomOneCS

Geometrically, (\AxiomOneCS) additionally requires the canonical inclusion $\partial\Delta_2 \hook \Delta_2$ from the boundary $\partial\Delta_2$ to the 2-simplex $\Delta_2$ is an equivalence, where $\partial\Delta_2$ type-theoretically is simply
\[ \partial\Delta_2 \coloneq \scomp{i,j : \I}{\isf{i} \vee (i = j) \vee \ist{i}}\text{.} \]

In general we will at most work with the weakest locality axiom (\AxiomNT). However, we will show in \zcref{sec:local} that $\I$, hence all spatial algebras and spectra, will be \emph{right orthogonal} to the maps that define these local principles as shown above. We will see in \zcref{sec:model} that these axioms corresponds to various coverages we can put on the presheaf classifying topos, and the orthogonality result particularly implies that $\I$ will belong to these subtopoi.

\section{Order-theoretic structure on algebras and spaces}\label{sec:order-theoretic-structure}

In this section, we describe several important local orderings on algebras and spaces, summarised in \zcref{table:orders}.

\begin{table}[ht]
  \begin{align*}
    a\le_A b &\eq a\wedge b = a
    \tag{canonical order}
    \\
    a\behle[A]b &\eq 
    \fa{x}{\spec{A}}x(a)\le_\I x(b)
    \tag{behavioural preorder}
    \\
    x\obsle[X]y &\eq \fa{U}{\opens{X}} U(x)\le_\I V(x)
    \tag{observational preorder}
    \\ 
    x\satle[\spec{A}]y &\eq \fa{a}{A} x(a)\le_\I y(a)
    \tag{satisfaction order}
  \end{align*}
  \begin{tabular}{ll}
  \end{tabular}
  \caption{List of important (pre)order structures on algebras and spaces. The \emph{canonical} and \emph{behavioural} (pre)orders are defined on all $\I$-algebras; the \emph{observational} preorder is defined on all types; and the \emph{satisfaction} preorder is defined just on spectra.}
  \label{table:orders}
\end{table}

\subsection{Local ordering of algebras}

Every $\I$-algebras has a built-in ``canonical'' partial ordering.

\begin{definition}[Canonical partial order]
  We shall write $\le_A$ for the \emph{canonical} partial ordering of an $\I$-algebra $A$ defined equivalently by meets or joins:
  \[ 
    a\le_A b \eq (a\wedge b = a) \eq (a\vee b = b)
  \]
\end{definition}

\begin{example}
  The canonical partial ordering of an algebra of observations $\opens{X}$ is given pointwise because $\opens{X}$ is the product of algebras $\prod_{x\in X}\I$:
  \begin{align*}
    U\le_{\opens{X}} V 
    &\eq \fa{x}{X} U(x)\le_\I V(x)
  \end{align*}
\end{example}

In addition to the canonical partial order, there is a local preordering of algebras that mirrors the observational order for spaces.

\begin{definition}[Behavioural preorder]
  For any $\I$-algebra $A$, we define the \emph{behavioural preorder} on $A$ as follows:
  \begin{align*}
    a \behle[A] b &\eq \fa{x}{\spec{A}} x(a) \le_{\I} x(b)
  \end{align*}
\end{definition}

To relate the behavioural preorder to the canonical partial ordering, we first define a slight weakening of the quasi-coherence principle.

\begin{definition}
  We say that an $\I$-algebra $A$ is \emph{pre-spatial} when the counit component $\iota_A\colon A\to \opens\spec{A}$ is an embedding.
\end{definition}

\begin{proposition}\label{lem:canon-beh-coincide}
  The canonical partial order of any pre-spatial $\I$-algebra $A$ is its behavioural preorder.
\end{proposition}

\begin{proof}
  For $a,b:A$ we have the following sequence of equivalences:
  \begin{align*}
    a\le_A b 
    &\eq \iota_A(a) \le_{\opens\spec{A}} \iota_A(a)
    \\ 
    &\eq \fa{x}{\spec{A}} x(a) \le_{\I} x(b)
    \\ 
    &\eq a\behle[A]b
  \end{align*}
  The first equivalence requires that $A$ is pre-spatial.
\end{proof}

\subsection{Local ordering of spaces}

The \emph{observational preorder} is a classical notion in synthetic domain theory (cf.\ \citet{PhoaWesleyKym-Son1991DtiR,hyland1990first}):

\begin{definition}[Observational preorder]\label{defn:specialisation}
  For any type $X$, we define the \emph{observational preorder} on $X$ as follows:
  \begin{align*}
    x\obsle[X] y 
    &\eq \fa{U}{\opens{X}} U(x)\le_\I U(x)
  \end{align*}
\end{definition}

By definition, the observational preorder is reflexive and transitive. As already observed in \emph{loc.\ cit.}, one important property of the observational preorder is that \emph{every} map is monotone w.r.t.\ this order:

\begin{lemma}\label{lem:anymapmonotoneintriscorder}
  For $f \colon X \to Y$, $x \obsle y$ in $X$ implies $fx \obsle fy$ in $Y$.
\end{lemma}
\begin{proof}
  This simply follows from compositionality of functions.
\end{proof}

We can very easily classify the observational preorder on a set with decidable equality.

\begin{lemma}[\AxiomNT]\label{lem:obs-preord-on-set-with-decidable-equality}
  If $M$ has decidable equality, then the observational preorder on $M$ is discrete, in the sense that for $n,m : M$,
  \[ n \obsle[M] m \to n = m\text{.} \]
\end{lemma}
\begin{proof}
  Since $M$ has decidable equality, we can define functions by case distinction. If $n \neq m$, we can construct a function $U : M \to \I$ with $U(n) = 1,U(m) = 0$, contradicting with $n \obsle[M] m$. Hence, $\neg\neg (n=m)$, and with decidable equality this implies $n = m$.
\end{proof}

\begin{corollary}[\AxiomNT]\label{cor:connectedpreservediscrete}
  If the observational preorder on $X$ has a least or greatest element, then $X$ is \emph{internally connected} in the sense that $X$ is right orthogonal to $M \to 1$ for any set $M$ with decidable equality.
\end{corollary}
\begin{proof}
  Since $X$ is inhabited, the restriction map $M \to M^X$ is an embedding (here we use the fact that $M$ is a set). It suffices to show that this map is also surjective.
  %
  Fixing $f\colon X\to M$, we must show that there exists $m:M$ such that $f(x) = m$ for all $x:X$. Let $x_0$ be the least element of $X$ in its observational preorder, so that we have $x_0\obsle[X] x$; by \zcref{lem:anymapmonotoneintriscorder} we have $f(x_0) \obsle f(x)$ which implies $f(x_0) = f(x)$ by \zcref{lem:obs-preord-on-set-with-decidable-equality}. Therefore, we may choose $m\coloneq f(x_0)$. The case for a maximal element is analogous.
\end{proof}

\begin{remark}
  With greater efforts, the above proof would also work with the weaker assumption that the observational preorder on $X$ is \emph{connected}. This generalises a similar result given by \citet[Prop.~4.4.1]{hyland1990first}.
\end{remark}

\subsection{The satisfaction order of a spectrum}

\zcref{cor:connectedpreservediscrete} demonstrated the usefulness of the observational preorder. We would like to use it to show e.g.\ that $\I$ is internally connected. For this, we need to characterise the observational preorder on sober spaces. It turns out that for any sober space, this is simply its \emph{satisfaction order} in the sense defined below:

\begin{definition}[Satisfaction order]
  We shall write $\satle[\spec{A}]$ for the \emph{satisfaction order} on the spectrum $\spec{A}$ of an $\I$-algebra as defined below:
  \begin{align*}
    x\satle[\spec{A}] y 
    &\eq \fa{a}{A}
    x(a) \le_\I y(a) 
  \end{align*}
\end{definition}

\begin{observation}\label{lem:specorderofsober}
  When $A$ is a spatial $\I$-algebra (and thus $\spec{A}$ is sober), the observational and satisfaction orders on $\spec{A}$ coincide.
\end{observation}
\begin{proof}
  For any $x,y:\spec{A}$ we have the following chain of equivalences:
  \begin{align*}
    x\obsle[\spec{A}] y 
    &\eq 
    \fa{U}{\opens\spec{A}} U(x) \le_{\I} U(y)
    \\ 
    &\eq 
    \fa{a}{A} 
    \iota_A(a)(x)
    \le_\I
    \iota_A(a)(y)
    \\ 
    &\eq 
    \fa{a}{A} 
    x(a)
    \le_\I
    y(a)
    \\ 
    &\eq 
    x\satle[\spec{A}] y
    \qedhere
  \end{align*}
\end{proof}

\begin{observation}
  The satisfaction order on the spectrum $\spec{\I[X]}$ of a free $\I$-algebra is induced by the canonical order on the \emph{observation algebra} $\opens{X}$ via the canonical bijection $\spec{\I[X]}\cong \opens{X}$.
\end{observation}

In the above, $\opens{X}$ is playing two roles simultaneously: it is both a spectrum of $\I[X]$ and the observational algebra of $X$.

\begin{proof}
  For any $u,v:\spec{\I[X]}$ we have the following equivalences:
  \begin{align*}
    u \satle[\spec{\I[X]}] v
    &\eq \fa{p}{\I[X]} u(p) \le_{\I} v(p)
    \\ 
    &\eq \fa{x}{X} u(x) \le_{\I} v(x)
    \\ 
    &\eq (\ld{x}{X}u(x)) \le_{\opens{X}} (\ld{x}{X}v(x))
    \qedhere
  \end{align*}
\end{proof}

\begin{lemma}\label{lem:cancoincide}
  For any surjective homomorphism $q \colon \I[Y] \surj A$ of a free algebra $\I[Y]$, the satisfaction order on $\spec A$ is induced by the satisfaction order on $\spec{\I[Y]}$ via the inclusion ${\spec{q}}\colon\spec{A}\hook \spec{\I[Y]}$.
\end{lemma} 

\begin{proof}
  Let $x,x' : \spec A$. We have the following equivalences:
  \begin{align*}
    (\spec{q})(x)\satle[\spec{\I[Y]}] (\spec{q})(x')
    &\eq 
    \fa{p}{\I[Y]}
    x(qp) \le_\I x'(qp)
    \\ 
    &\eq 
    \fa{a}{A} x(a) \le_\I x'(a) 
    \\ 
    &\eq 
    x\satle[\spec{A}] x'
  \end{align*}
  The second equivalence holds because $q\colon \I[Y]\surj A$ is surjective.
\end{proof}


\section{Phoa's principle, finitary quasi-coherence, and homotopy}

The interval $\I$ is playing multiple roles so far:
\begin{enumerate}
  \item Mapping a space $X$ into $\I$ computes its algebra of observations $\opens{X} = \I^X$.
  \item Homomorphically mapping an algebra $A$ into $\I$ computes its spectrum $\spec{A} = \I\alg(A,\I)$.
\end{enumerate}

In both cases above, $\I$ is viewed as an algebra of observations and not as a geometrical figure. For the latter, we may consider mappings from $\I$ into \emph{either} an algebra or a space as defining a notion of (directed) homotopy. In particular, the function space $X^\I$ classifies \emph{paths} drawn in $X$; this geometrical role of $\I$ as a figure shape for paths is the primary one in applications to synthetic category theory~\citep{riehl2017type}, and it is also upon the geometry of the interval that the notion of \emph{chain completeness} in synthetic domain theory is founded (cf.\ \zcref{sec:infdomain}).

\subsection{The generalised Phoa principle and spatiality of polynomial algebras}

As soon as we have path spaces $X^\I$, we immediately wish to begin characterising them for specific $X$. For example, the path space of a function space $A\to B$ is given pointwise in terms of the path space of $B$, which follows immediately from the laws of exponentials. To characterise paths in spaces like $\I^n$ and $\Delta^n$, however, we need additional axioms. The simplest such axiom asserts that the polynomial algebra $\I[\ms{i}]$ is spatial---which, in the case of bounded distributive lattices, is equivalent to Phoa's principle.

\begin{definition}\label{def:gen-phoa}
  We say that an $\I$-algebra $A$ satisfies the \emph{generalised Phoa principle} when either of the following equivalent conditions hold:
  \begin{enumerate}
    \item The path space $A^\I$ classifies the canonical order on $A$.
    \item For any function $\alpha\colon \I\to A$, we have $\alpha(i) = \alpha(0)\vee i \wedge \alpha(1)$.
  \end{enumerate}
\end{definition}

\begin{proof}
  Suppose that $A^\I$ classifies the canonical order on $A$; we must show that for any $\alpha\colon \I\to A$ we have $\alpha(i) = \alpha(0)\vee i \wedge \alpha(1)$. By our assumption that $A^\I$ classifies the canonical order on $A$, we know that that any such function must be uniquely determined by its values on $0$ and $1$ and moreover that $\alpha(0) \le_A \alpha(1)$; therefore, it suffices to observe for arbitrary such $\alpha$ that
  \begin{align*}
    \alpha(0) &= \alpha(0) \vee 0 = \alpha(0)\vee 0 \wedge \alpha(1)\text{,}\\
    \alpha(1) &= \alpha(0)\vee\alpha(1) = \alpha(0)\vee 1 \wedge \alpha(1)\text{.}
  \end{align*}

  Conversely, suppose that every path $\alpha\colon \I\to A$ is of the form $\alpha(i) = \alpha(0)\vee i \wedge \alpha(1)$. We must show that $a\le_A b$ holds if and only if there exists a unique path $\alpha\colon \I\to A$ sending $0$ and $1$ to $a$ and $b$ respectively.
  \begin{enumerate}
    \item If $a\le_A b$ holds, we define $\alpha(i) \coloneq a\vee i \wedge b$; this path is unique with the required property by our assumption. 
    \item If there exists a path $\alpha\colon\I\to A$ from $a$ to $b$, we know that this takes the form $\alpha(i) = a \vee i \wedge b$. Therefore we have
    \[ 
      b = \alpha(1) = a\vee 1\wedge b = a \vee b.
      \qedhere
    \] 
  \end{enumerate}
\end{proof}

The standard Phoa principle is the generalised Phoa principle for $\I$ itself.

\begin{theorem}\label{thm:gen-phoa-poly}
  If $A$ satisfies the generalised Phoa principle, then so does the polynomial $\I$-algebra $A[\ms{j}]$.
\end{theorem}

\begin{proof}
  Fixing $\alpha\colon \I\to A[\ms{j}]$, we must check that $\alpha(i) = \alpha(0) \vee i \wedge \alpha(1)$. 
  The normal form theorem for polynomials of bounded distributive lattices~\citep[Ch.\ 1, Thm.\ 10.11]{lausch2000algebra} implies that any function $\alpha\colon \I\to A[\ms{j}]$ must be determined by functions $\alpha_0,\alpha_1\colon \I\to A$ such that $\alpha(i) = \alpha_0(i)\vee \ms{j} \wedge \alpha_1(i)$ in $A[\ms{j}]$. Meanwhile, Phoa's principle for $A$ characterises $\alpha_0$ and $\alpha_1$ as follows:
  \begin{align*}
    \alpha_0(i) &= \alpha_0(0) \vee i \wedge \alpha_0(1)
    \\ 
    \alpha_1(i) &= \alpha_1(0) \vee i \wedge \alpha_1(1)
  \end{align*}

  Hence we have $\alpha(i) = (\alpha_0(0) \vee i \wedge \alpha_0(1))\lor x \wedge (\alpha_0(0) \vee i \wedge \alpha_0(1))$. On the other hand, we can use the normal form to compute $\alpha(0)$ and $\alpha(1)$ as polynomials in $\ms{j}$: 
  \begin{align*}
    \alpha(0) &= \alpha_0(0) \vee \ms{j}\wedge \alpha_1(0)\\
    \alpha(1) &= \alpha_0(1) \vee \ms{j}\wedge \alpha_1(1)
  \end{align*}

  Hence $\alpha(0)\lor i \wedge \alpha(1) = (\alpha_0(0) \vee \ms{j}\wedge \alpha_1(0))\vee i \wedge (\alpha_0(1) \vee \ms{j}\wedge \alpha_1(1))$. Using elementary lattice algebra we deduce that $\alpha(i) = \alpha(0) \vee i \wedge \alpha(1)$.
\end{proof}

\begin{theorem}\label{obs:sqcp-phoa-gen}
  Let $A$ be a spatial $\I$-algebra. Then the following are equivalent:
  \begin{enumerate}
    \item The polynomial $\I$-algebra $A[\ms{j}]$ is spatial.
    \item The generalised Phoa principle holds for $A$.
  \end{enumerate}
\end{theorem}

\begin{proof}
  We note that $\spec{A[\ms{j}]} \cong \spec{A}\times \I$ unconditionally, and so the algebra of observations $\opens\spec{A[\ms{j}]}$ is actually the path space of $\opens\spec{A}$; taking spatiality of $A$ into account, the counit component $A[\ms{j}]\to \opens\spec A[\ms{j}]$ is the evaluation map $A[\ms{j}] \to A^\I$, and we wish to prove that the latter is an equivalence if and only if for each $\alpha\colon \I\to A$ we have $\alpha(i)=\alpha(0)\vee i\wedge\alpha(1)$. 
  This follows immediately from the normal form theorem for polynomials of bounded distributive lattices~\citep[Ch.\ 1, Thm.\ 10.11]{lausch2000algebra}: any polynomial $p:A[\ms{j}]$ is uniquely of the form $p = \ev_0(p) \vee \ms{j} \wedge \ev_1(p)$ with $\ev_0(p)\le_A \ev_1(p)$.
\end{proof}

\begin{corollary}\label{cor:A-and-Ax-qc-to-Axy-qc}
  If both $A$ and $A[\ms{j}]$ are spatial, then so is $A[\ms{j},\ms{k}]$.
\end{corollary}
\begin{proof}
  By \zcref{obs:sqcp-phoa-gen}, the generalised Phoa principle holds for $A$ because $A$ and $A[\ms{j}]$ are spatial; by \zcref{thm:gen-phoa-poly}, the generalised Phoa principle holds for $A[\ms{j}]$ because it holds for $A$.  By \zcref{thm:gen-phoa-poly} again, the generalised Phoa principle holds for $A[\ms{j},\ms{k}] = A[\ms{j}][\ms{k}]$ because it holds for $A[\ms{j}]$.
\end{proof}

\begin{corollary}\label{cor:phoa-vs-quasicoherence}
  The following are equivalent:
  \begin{enumerate}
    \item The polynomial $\I$-algebra $\I[\ms{i}]$ is spatial.
    \item Phoa's principle holds.
  \end{enumerate}
\end{corollary}

\begin{proof}
  By \zcref{obs:sqcp-phoa-gen}, because $\I$ is unconditionally spatial.
\end{proof}

We therefore identify the axiom (\AxiomSQCP) below accordingly, which we have seen to be equivalent to Phoa's principle.

\PrintAxiomSQCP

The quasi-coherence axiom for polynomials in one variable is surprisingly strong:



\begin{lemma}[\AxiomSQCP]\label{lem:fg-qc}
  Every finitely generated free $\I$-algebra is spatial.
\end{lemma}

\begin{proof}
  The free $\I$-algebra on $0$ generators is just $\I$, which is automatically spatial. The free $\I$-algebra $\I[\ms{i}]$ on one generator is spatial by assumption. In the case of $\I[\ms{i},\ms{j}_0,\ldots,\ms{j}_n]$ we recall \zcref{cor:A-and-Ax-qc-to-Axy-qc} and see that $\I[\ms{i}]$ spatial implies $\I[\ms{i},\ms{j}_0]$ spatial, which implies $\I[\ms{i},\ms{j}_0,\ms{j}_1]$ spatial and so on.
\end{proof}

\begin{theorem}[\AxiomSQCP]\label{thm:phoa-spectrum}
  For any $\I$-algebra $A$, the path space ${(\spec A)}^\I$ classifies the satisfaction order on $\spec A$.
\end{theorem}
\begin{proof}
  By assumption $\I[\ms{i}]$ is spatial, so from \zcref{prop:duality} may characterise the path space of $\spec{A}$ algebraically as follows:
  \[ (\spec A)^\I \cong (\spec{A})^{\spec{\I[\ms{i}]}}\cong \I\alg(A,\I[\ms{i}])\text{.} \]
  
  By the normalisation theorem for $\I[\ms{i}]$, the pair of evaluation maps
  \[ \pair{\ev_0,\ev_1} \colon \I[\ms{i}] \to \I \times \I \]
  classifies the canonical order on $\I$.
  Thus every homomorphism $h\colon A\to \I[\ms{i}]$ is canonically induced by elements $h_0,h_1:\spec{A}$ with $h_0\satle[\spec{A}] h_1$ such that $h(a) = h_0(a) \vee \ms{i} \wedge h_1(a)$.
\end{proof}

\begin{theorem}[\AxiomSQCP]\label{thm:algebraphoa}
  If both $A$ and $A[\ms{i}]$ are spatial $\I$-algebras, then the path space $A^\I$ classifies the canonical order on $A$.
\end{theorem}
\begin{proof}
  By assumption and \zcref{specrightadj}, since $A \otimes \I[\ms{i}] \cong A[\ms{i}]$, we have
  \[ A^\I \cong (\opens\spec A)^\I = (\I^{\spec A})^\I \cong \I^{\spec A \times \spec\I[\ms{i}]} \cong \I^{\spec A[\ms{i}]} \cong A[\ms{i}]\text{.} \]
  We have already seen that $A[\ms{i}]$ classifies the canonical order of $A$ by the normalisation result for polynomial $\I$-algebras.
\end{proof}

\begin{remark}[Algebraic properties in classifying topoi]\label{rem:normalalgebra}
  We emphasise that the above proofs are purely the consequences of the algebraic fact that $A[\ms{i}]$ classifies the canonical order of any $\I$-algebra $A$, plus quasi-coherence. This is a perfect example of how an algebraic property of a theory has a non-trivial effect on the internal logic of its classifying topos.
\end{remark}

\begin{corollary}[\AxiomSQCP]\label{lem:cubes-are-sober}
  Each cube $\I^n$ is sober.
\end{corollary}
\begin{proof}
  The $n$-cube $\I^n$ is the spectrum of the free $\I$-algebra $\I[n]$, which is spatial by \zcref{lem:fg-qc}.
\end{proof}

\begin{corollary}[\AxiomNT, \AxiomSQCP]\label{cor:cubes-internally-connected}
  The $n$-cube $\I^n$ is internally connected.
\end{corollary}

In particular, $\I$ is not the boolean set $2$.

\begin{proof}
   By \zcref{lem:specorderofsober,lem:cancoincide}, the observational preorder space $\I^n$ is induced by the satisfaction order on $\spec\I[n]$ because $\I^n$ is sober (\zcref{lem:cubes-are-sober}); the satisfaction order is simply the pointwise order, so it has a top element. Thus by \zcref{cor:connectedpreservediscrete}, $\I^n$ is internally connected. So are $\Delta^n$ because they are retracts of $\I^n$.
\end{proof}

\subsection{Comparing the different orders}

An immediate application of Phoa principle is to compare path structure with the observational preorder.

\begin{observation}[\AxiomSQCP]
  For any type $X$, the boundary evaluation map 
  \[ \pair{\ev_0,\ev_1}\colon X^\I\to X\times X \] 
  factors through the observational preorder relation ${\obsle[X]}\hook X\times X$ as displayed below:
  \[ 
    \begin{tikzcd}
      X^\I 
        \arrow[rr,densely dashed, "\exists"]
        \arrow[dr]
      && 
      {\obsle[X]}
        \arrow[dl,hookrightarrow]
      \\ 
      &
      X\times X
    \end{tikzcd}
  \] 
\end{observation}
\begin{proof}
  Fix a path $\alpha\colon \I\to X$. We have $0\obsle[\I] 1$ by Phoa's principle; because every function is monotone in the observation preorder, we conclude $\alpha(0)\obsle[X] \alpha(1)$.
\end{proof}

\begin{definition}
  A type $X$ for which $X^\I\to {\obsle[X]}$ is an surjection is called \emph{linked} (following \citet{PhoaWesleyKym-Son1991DtiR}); when the map is an equivalence, $X$ is called \emph{strongly linked}.
\end{definition}

\begin{lemma}[\AxiomSQCP]\label{lem:sober-strongly-linked}
  Any sober space is strongly linked.
\end{lemma}

\begin{proof}
  Let $X=\spec{A}$ be the spectrum of a spatial algebra $A$. 
  By \zcref{thm:phoa-spectrum}, the path space $X^\I$ classifies the satisfaction order on $\spec{A}$, which is to say that there is a (necessarily unique) path from $x$ to $y$ if and only if \[\fa{a}{A} x(a)\le_\I y(a)\text{.}\] Because $A$ is spatial, we have $A\cong \opens X$ and there is a unique path from $x$ to $y$ if and only if \[\fa{U}{\opens{X}} U(x)\le_\I U(y)\text{,}\] but this is the observational order of $X$.
\end{proof}

\begin{lemma}[\AxiomSQCP]\label{lem:int-strong-link}
  The interval is strongly linked.
\end{lemma}

\begin{proof}
  The interval is sober because it is the spectrum of $\I[\ms{i}]$, which is spatial by (\AxiomSQCP).
\end{proof}

\begin{lemma}\label{lem:strong-link-exp-ideal}
  Being strongly linked is an exponential ideal.
\end{lemma}

\begin{proof}
  \citet[Prop.\ 5.4.4]{PhoaWesleyKym-Son1991DtiR} implies that being linked is an exponential ideal, because an exponential is an internal limit. A type $X$ is strongly linked if and only if it is linked \emph{and} it is $\I$-separated in the sense that $X^\I\to X\times X$ is an embedding; but $\I$-separated types also form an exponential ideal (in fact, a reflective exponential ideal).
\end{proof}

\begin{lemma}[\AxiomSQCP]\label{lem:O-strongly-linked}
  Any algebra of observations is strongly linked.
\end{lemma}

\begin{proof}
  We assume that $A=\opens{X}$ for some not necessarily sober $X$; then $\opens{X}=\I^X$ is strongly linked by \zcref{lem:strong-link-exp-ideal,lem:int-strong-link}.
\end{proof}

\begin{corollary}\label{cor:strong-link-obs-pw}
  When $Y$ is strongly linked, the observational preorder on $Y^X$ is induced pointwise by the observational preorder on $Y$.
\end{corollary}

\begin{proof}
  Let $Y$ be strongly linked; by \zcref{lem:strong-link-exp-ideal}, so is $Y^X$. Thus $\obsle[Y^X]$ is precisely $(Y^X)^\I\cong (Y^\I)^X\cong (\obsle[Y])^X$.
\end{proof}

\begin{proposition}[\AxiomSQCP]\label{prop:specordopens}
  On any observation algebra $\opens{X}$, the observational preorder coincides with the canonical order.
\end{proposition}

\begin{proof}
  By Phoa's principle, $\I$ is strongly linked; thus by \zcref{cor:strong-link-obs-pw}, the observational preorder on $\I^X$ is pointwise induced by that of $\I$. In particular, we have the following equivalences involving the function space $\opens{X}=\I^X$:
  \begin{align*}
    U\obsle[\I^X]V 
    &\eq \fa{x}{X} U(x)\obsle[\I] V(x)
    \\
    &\eq \fa{x}{X} U(x)\le_{\I} V(x)
    \\
    &\eq U\le_{\opens{X}} V \qedhere
  \end{align*}
  The second equivalence uses Phoa's principle and \zcref{lem:O-strongly-linked} to identify the canonical order on $\I$ with its observational preorder.
\end{proof}

\begin{proposition}[\AxiomSQCP]\label{cor:specordonalgiscan}
  For a spatial $\I$-algebra $A$, the observational preorder on the underlying set of $A$ is induced by the canonical order on $\opens\spec{A}$ by the counit isomorphism $\iota_A\colon A\to \opens\spec{A}$. Hence the observational order and canonical order of $A$ coincide.
\end{proposition}

\begin{proof}
  Because $A$ is spatial, the counit component $\iota_A$ gives us an order-isomorphism ${(A,\le_A)}\to (\opens\spec{A},\le_{\opens\spec{A}})$. By \zcref{prop:specordopens}, this is an order-isomorphism ${(A,\le_A)}\to {(\opens\spec{A},\obsle[\opens\spec{A}])}$. Meanwhile, any bijective function tracks an isomorphism of observational preorders, so the inverse function $\iota_A^{-1}\colon \opens\spec{A}\to A$ tracks an order-isomorphism  ${(\opens\spec{A},\obsle[\opens\spec{A}])}\to {(A,\obsle[A])}$. 
\end{proof}

\begin{corollary}[\AxiomSQCP]
  The canonical partial order, observational preorder, and behavioural preorder of a spatial $\I$-algebra all coincide.
\end{corollary}

\begin{proof}
  By \zcref{cor:specordonalgiscan,lem:canon-beh-coincide}.
\end{proof}

\subsection{Spatiality of finitely presented algebras}

The (\AxiomSQCP) axiom does not imply that the simplices $\Delta^n$ are sober; for this, we need a stronger quasi-coherence principle that applies to all finitely \emph{presented} algebras.

\PrintAxiomSQCF

\begin{remark}
  In particular, the axiom (\AxiomSQCF) implies all our previously mentioned quasi-coherence axioms, including (\AxiomSQCP) and (\AxiomSQCID) (hence also (\AxiomSQCI)).
\end{remark}

It is immediate that $\Delta^n\cong \spec\I[\ms{i}_1,\cdots,\ms{i}_n]/\ms{i}_1\geq \cdots \geq \ms{i}_n$ is sober under (\AxiomSQCF).

\begin{corollary}[\AxiomSQCF]
  Each $n$-simplex $\Delta^n$ is internally connected.
\end{corollary}

\begin{proof}
  This follows in the same way as \zcref{cor:cubes-internally-connected}.
\end{proof}

\subsection{Classification of simplices by algebras}

At the end of this section, we describe another interesting perspective arising from the proof of \zcref{thm:phoa-spectrum}. We have seen that the dualising object $\I$ has a double role: It is both an algebra and a spectrum. The mixture of algebraic and geometric object has also been observed in \zcref{rem:alggeoI}. The proof of \zcref{thm:phoa-spectrum} gives us many more such examples. For instance, $\I[i]$ classifies the order on $\I$, which by definition is the spectrum $\Delta_2$. In fact, \emph{all} the simplices are classified by some algebra.

Just as we described $\Delta^n$ in terms of descending sequences in $\I$, we can do the same replacing $\I$ with another algebra $A$.
For any $n : \N$ and $\I$-algebra $A$, we define $\Delta[A]^{n}$ to be the type of lists of decreasing elements of length $n$ in $A$, namely
\[ \Delta[A]^{n} \coloneq \scomp{a_1,\cdots,a_n : A}{a_1 \ge \cdots \ge a_n}\text{.} \]
Of course $\Delta[\I]^n$ is simply $\Delta^n$. We have a general algebraic description of $\Delta[A]^n$.

\begin{proposition}\label{prop:simplicesasalgebra}
  For any $n : \N$ and $\I$-algebra $A$, there is an equivalence 
  \[ \Delta[A]^{n+1} \cong A[n]/\ms{i}_1 \le \cdots \le \ms{i}_n\text{,} \]
  sending any $a_0 \ge \cdots \ge a_n$ in $A$ to the polynomial 
  \[ a_0 \wedge \ms{i}_1 \vee a_1 \wedge \cdots \wedge \ms{i}_n \vee a_n\text{.} \]
\end{proposition}

We omit the proof here, which again follows from a general normal form for polynomials with finitely many variables for bounded distributive lattices; see \citet[Ch.\ 1, Thm.\ 10.21]{lausch2000algebra}. We only note here that the above polynomial is well-defined, in the sense that the expression is associative and its value does not depend on how one add parenthesises. The reason for this is that for any expression $a \vee b \wedge c$ to have a definite meaning in a distributive lattice, i.e.\ the two value $(a \vee b) \wedge c$ and $a \vee (b \wedge c)$ coincide, it suffices to have $a \le c$. Then the condition $a_0 \ge \cdots \ge a_n$ for the parameters, and the fact that $\ms{i}_1 \le \cdots \le \ms{i}_n$, makes sure the above polynomial has a definite meaning in the quotient.

\section{Chain completeness and infinitary domain theory}\label{sec:infdomain}

Until this point, we have seen that elementary axioms for synthetic domain theory follow from (SQCF) for bounded distributive lattices and thus $\sigma$-frames. In this section we show that the final axiom of synthetic domain theory, viz.\ chain completeness of the interval $\I$, is also a consequence of quasi-coherence for $\sigma$-frames. In fact, it implies the main infinitary axiom of synthetic domain theory in~\citet{fiore-plotkin:1996}---the initial algebra $\omega$ for $L$ is inductive.

\subsection{Chain completeness and inductivity of the initial algebra}

Chain completeness is again specified as an orthogonality condition. Let $\omega,\ov\omega$ denote the initial algebra and final coalgebra for the partial map classifier $L$, respectively. Intuitively, $\omega$ behaves as an infinite chain, and $\ov\omega$ as an infinite chain with a top element. The chain completeness condition is thus intuitively stating a type has joins for infinite sequences internally:

\begin{definition}[Chain completeness]
  A type $X$ is \emph{chain complete} if it is right orthogonal to the inclusion $\omega \hook \ov\omega$.
\end{definition}

The importance for chain completeness of $\I$ for domain theory is that, as observed in~\citet{hyland1990first}, it produces fixed points of endo-morphisms on a suitably defined class of objects, which could be viewed as \emph{domains}. 

As $L$ by construction is a polynomial functor, it preserves connected limits. This implies the final coalgebra $\ov\omega$ can always be constructed as a sequential limit as follows, 
\[\begin{tikzcd}
	{\ov\omega\cong\lt_{n:\N}L^n1} \ar[r] & \cdots & {L^21} & L1 & 1
	\arrow[from=1-2, to=1-3]
	\arrow["{{L!}}"', from=1-3, to=1-4]
	\arrow["{{!}}"', from=1-4, to=1-5]
\end{tikzcd}\]
However, the dual statement fails for the initial algebra $\omega$, i.e.\ it is \emph{not} always the following sequential colimit,
\[\begin{tikzcd}
	\emp & {L(\emp)} & {L^2(\emp)} & \cdots
	\arrow["{{?}}"', from=1-1, to=1-2]
	\arrow["{{L(?)}}"', from=1-2, to=1-3]
	\arrow[from=1-3, to=1-4]
\end{tikzcd}\]

\begin{definition}
  $\omega$ is \emph{inductive} if it is the sequential colimit above.
\end{definition}
The failure of inductivity of $\omega$ has been observed in various realisability models for synthetic domain theory~\citep{VANOOSTEN2000233}. As shown by \citet{fiore-plotkin:1996}, the inductivity of $\omega$ is indeed a desirable property: If it holds, then there will be a much closer correspondence between models for \emph{synthetic domain theory} and models for \emph{axiomatic domain theory}. Specifically, for a model of synthetic domain theory where $\omega$ is inductive, one can construct a corresponding model of axiomatic domain theory, where the domains are the \emph{well-complete types}\footnote{A type $X$ is well-complete if $LX$ is chain complete.}. The inductivity of $\omega$ makes $\ov\omega$ an \emph{inductive fixed point object} in the category of domains in the sense of \emph{loc.\ cit.}, which is part of their requirement for an axiomatic model of domain theory. 

The remaining part of this section is dedicated to showing that under a suitable quasi-coherence assumption, we can indeed show that $\I$ is chain complete (\zcref{thm:complete}) and $\omega$ is inductive (\zcref{prop:omegacolimit}). In fact, the proof of \zcref{thm:complete} is the only place we where have used the essential properties of $\sigma$-frames, and we will also indicate why it fails for bounded distributive lattices in \zcref{rem:whynotdis}.

\subsection{Spatiality of countably presented algebras and infinitary domain theory}
The connection between the infinitary aspect of synthetic domain theory and quasi-coherence starts from the observation that the final coalgebra $\ov\omega$ for the functor $L$ can be described as a \emph{spectrum}:

\begin{example}[$\ov\omega$ is a spectrum]\label{exm:ovomegasober}
  As mentioned, the final coalgebra $\ov\omega$ can be characterised as the sequential limit,
  \[\begin{tikzcd}
    \cdots & {\Delta^2} & {\Delta^1} & {\Delta^0} \\
    \cdots & {L^21} & L1 & 1
    \arrow[from=1-1, to=1-2]
    \arrow[from=1-2, to=1-3]
    \arrow["\cong"{description}, from=1-2, to=2-2]
    \arrow[from=1-3, to=1-4]
    \arrow["\cong"{description}, from=1-3, to=2-3]
    \arrow["\cong"{description}, from=1-4, to=2-4]
    \arrow[from=2-1, to=2-2]
    \arrow["{L!}"', from=2-2, to=2-3]
    \arrow["{!}"', from=2-3, to=2-4]
  \end{tikzcd}\]
  where under the isomorphisms, the transition map $\Delta^{n+1} \to \Delta^n$ takes the sequence $i_0 \ge \cdots \ge i_n$ to the final segment $i_1 \ge \cdots \ge i_n$. As observed by \citet[Sec.\ 5.2]{hyland1990first}, it also has an equivalent type-theoretic description as the object of infinite descending sequences in $\I$,
  \[ \ov\omega \cong \scomp{i : \N \to \I}{\fa n\N i_n \ge i_{n+1}}\text{.} \]
  This way, we have a natural equivalence
  \[ \ov\omega \cong \spec\I[\N]/\pair{\ms{i}_n \ge \ms{i}_{n+1}}_{n:\N}\text{.} \]
  Here now $\I[\N]/\pair{\ms{i}_n \ge \ms{i}_{n+1}}_{n:\N}$ is a \emph{countably presented} $\I$-algebra.
\end{example}

In general, a countably presented $\I$-algebra is of the form $\I[\N]/\pair{s_n = t_n}_{n:\N}$ for some $\N$-indexed lists of terms $s,t \colon \N \to \I[\N]$. In particular, all finitely presented $\I$-algebras will also be countably presented. Motivated by the above characterisation of $\ov\omega$, we naturally consider the following stronger quasi-coherence principle:

\PrintAxiomSQCC

\begin{remark}
  Our modular development in the previous sections can be now immediately applied to $\ov\omega$ if we assume (\AxiomSQCC). For instance, the description of the observational preorder for sober spaces in \zcref{lem:specorderofsober} now applies to $\ov\omega$, which shows this again coincides with its pointwise order as a subspace of $\I^\N$. In particular, it also has both a top and bottom element, thus \zcref{cor:connectedpreservediscrete} now implies $\ov\omega$ is also internally connected. 
\end{remark}

(\AxiomSQCC) combined with the non-triviality axiom (\AxiomNT) has many logical consequences. The crucial observation is that (\AxiomNT) implies a weak form of \emph{Nullstellensatz} result, as already noted by several authors~\citep{blechschmidt2021using,blechschmidt2020general,Cherubini_Coquand_Hutzler_2024}:

\begin{lemma}[\AxiomNT]\label{lem:nulls}
  For an sober space $X \cong \spec A$, $X \cong \emp$ iff $A$ is trivial, viz.\ $0=1$ in $A$.
\end{lemma}
\begin{proof}
  The backward direction holds since because $\I$ is non-trivial (\AxiomNT), thus there is no homomorphism from a trivial algebra to $\I$. For the forward direction: by assumption $A \cong \I^{\spec A} \cong \I^\emp$, which implies $A$ is trivial. 
\end{proof}

Together with (\AxiomSQCC), this implies the following form of Markov's principle; a similar result is shown by \citet{cherubini2024foundation}:

\begin{lemma}[\AxiomNT, \AxiomSQCC]\label{lem:markov}
  For any $i : \ov\omega$, we have
  \[ \neg\fa{n}{\N}\ist{i_n} \to \ex n\N\isf{i_n}\text{.} \]
\end{lemma}
\begin{proof}
  Let $i : \ov\omega$. Note that similar to \zcref{lem:openpropsober}, the proposition $\fa n\N \ist{i_n}$ by construction is the following sober space, 
  \[ \spec\I/i \cong \I\alg(\I/i,\I) \cong \fa n\N \ist{i_n}\text{,} \]
  where we have abbreviated the countably presented $\I$-algebra $\I/\pair{i_n=1}_{n:\N}$ as $\I/i$. Now if we have $\neg\fa n\N \ist{i_n}$, then $\spec\I/i \cong \emp$ which by \zcref{lem:nulls} implies $\I/i$ is trivial. But this algebra is trivial iff $\ex n\N \isf{i_n}$.
\end{proof}

Equipped with the above result, we can now proceed to study the initial algebra $\omega$ for $L$. \citet{JIBLADZE1997185} has given a beautiful formula for a type-theoretic description of $\omega$ as the following subset of $\ov\omega$, 
\[ \omega \coloneq \scomp{i : \ov\omega}{\fa\phi{\pp} (\fa n{\N} (\ist{i_n} \to \phi) \to \phi) \to \phi}\text{.} \]
For another proof, see e.g.\ that of \citet{VANOOSTEN2000233}. In the presence of \zcref{lem:markov}, this description can be greatly simplified:

\begin{proposition}[\AxiomNT, \AxiomSQCC]\label{prop:omegacolimit}
  $\omega$ is equivalent to the following subset of $\ov\omega$,
  \[ \omega \cong \scomp{i : \ov\omega}{\ex n{\N} \isf{i_n}}\text{.} \]
  In particular, $\omega$ is the colimit of the following sequence, thus is inductive:
  \[\begin{tikzcd}
    \emp & {\Delta^0} & {\Delta^1} & \cdots \\
    \emp & {L(\emp)} & {L^2(\emp)} & \cdots
    \arrow[from=1-1, to=1-2]
    \arrow["\cong"{description}, from=1-1, to=2-1]
    \arrow[from=1-2, to=1-3]
    \arrow["\cong"{description}, from=1-2, to=2-2]
    \arrow[from=1-3, to=1-4]
    \arrow["\cong"{description}, from=1-3, to=2-3]
    \arrow["{?}"', from=2-1, to=2-2]
    \arrow["{L(?)}"', from=2-2, to=2-3]
    \arrow[from=2-3, to=2-4]
  \end{tikzcd}\]
\end{proposition}

The transition maps $\Delta^n\to \Delta^{n+1}$ above are given by appending $0$ to the end of a descending sequence. Under the identification $\Delta^{n+1}\cong T(\Delta^n)$, the corresponding map $\Delta^n\to T(\Delta^n)$ is the unit component $\eta_T\colon \Delta^n\to T(\Delta^{n})$.

\begin{proof}
  Let $i : \ov\omega$. It suffices to show that
  \[ \prth{\fa\phi{\pp} (\fa n{\N} (\ist{i_n} \to \phi) \to \phi) \to \phi} \to \ex{n}\N \isf{i_n}\text{.} \]
  Assume the premise. We can instantiate $\phi$ to $\emp$. By assumption, $\neg\fa n\N \neg\neg\ist{i_n}$, which by \zcref{cor:opendnegclose} is equivalent to $\neg\fa n\N \ist{i_n}$. Then \zcref{lem:markov} implies this is $\ex n\N \isf{i_n}$. The fact that this makes $\omega$ into the above sequential colimit follows from \citet[Cor.~1.10]{VANOOSTEN2000233}.
\end{proof}

  

Using this colimit description of $\omega$ and the fact that $\ov\omega$ is sober, we can show $\I$ is chain complete. This again crucially depends on a normal form result for a countably presented $\sigma$-frame, which generalises the finitary version for bounded distributive lattices given in \zcref{prop:simplicesasalgebra}:

\begin{lemma}\label{app:normalsigma}
  For the countably presented $\sigma$-frame $\I[\N]/\pair{\ms{i}_n \ge \ms{i}_{n+1}}_{n:\N}$, an element $p$ can be uniquely written as 
  \[ p = p_0 \vee \bigvee_{n:\N} p_{n+1} \wedge \ms{i}_n\text{,} \]
  with $p_n \le p_{n+1}$ for all $n$. In other words, $\I[\N]/\pair{\ms{i}_n \ge \ms{i}_{n+1}}_{n:\N}$ is isomorphic to the following $\sigma$-frame with the pointwise order induced by $\I$,
  \[ \Delta_\infty \coloneq \scomp{i : \N \to \I}{\fa n\N i_n \le i_{n+1}}\text{.} \]
\end{lemma}
\begin{proof}
  We directly prove that $\Delta_\infty$ satisfies the universal property of the countably presented $\sigma$-frame $\I[\N]/\pair{\ms{i}_n \ge \ms{i}_{n+1}}_{n:\N}$. We pick the generators in $\Delta_\infty$ as follows,
  \[ i_n \coloneq (\underbrace{0,\cdots,0}_{n+1 \text{ times}},1,1,\cdots)\text{.} \]
  For any $\sigma$-frame $A$ with $a_0 \ge a_1 \ge \cdots$, we define a map $f_a : \Delta_\infty \to A$ sending $j = (j^0,j^1,\cdots) : \Delta_\infty$ to
  \[ f_a(i) \coloneq j_0 \vee \bigvee_{n:\N} j^{n+1} \wedge a_{n}\text{.} \]
  By construction it is easy to see $f_a$ is a $\sigma$-frame morphism. Evidently for any $n : \N$, 
  \[ f_a(i_n) = \bigvee_{m:\N} i_n^{m+1} \wedge a_m = a_n\text{.} \]
  Furthermore, for any $\sigma$-frame map $f \colon \Delta_\infty \to A$ with $f(i_n) = a_n$, we must have $f = f_a$ because any $j$ in $\Delta_\infty$ can be written as
  \[ j = j^0 \vee \bigvee_{n:\N} j^{n+1} \wedge i_n\text{,} \]
  which implies that
  \[ f(j) = j^0 \vee \bigvee_{n:\N}j^{n+1} \wedge a_n = f_a(j)\text{.} \]
  This completes the proof.
\end{proof}

\begin{theorem}[\AxiomNT, \AxiomSQCC]\label{thm:complete}
  $\I$ is right orthogonal to the inclusion $\omega\hook\ov\omega$.
\end{theorem}
\begin{proof}
  Since $\ov\omega$ is now sober, we have
  \[ \I^{\ov\omega} \cong \I[\N]/\pair{\ms{i}_n \ge \ms{i}_{n+1}}_{n:\N}\text{.} \]
  On the other hand, since $\omega$ is the colimit of $\Delta^n$ and they are sober, we have
  \[ \I^\omega \cong \lt_{n:\N}\I^{\Delta^n} \cong \lt_{n:\N}\I[n]/\ms{i}_0\ge\cdots\ge \ms{i}_{n-1}\text{.} \]
  Note that the transition maps induced by $\eta_T \colon \Delta^n \to T(\Delta^n)\cong \Delta^{n+1}$ under quasi-coherence gives us the following maps on algebras:
  \[
  \begin{tikzcd}
    \I^{\Delta^{n+1}} \ar[r, "\I^{\eta_T}"] & \I^{\Delta^n} \\ 
    \I[n\!+\!1]/\ms{i}_0 \ge \cdots \ge \ms{i}_{n} \ar[u, "\cong"] \ar[r, "\ms{i}_{n} \mapsto 0"'] & \I[n]/\ms{i}_1 \ge \cdots \ge \ms{i}_{n-1} \ar[u, "\cong"']
  \end{tikzcd}
  \]
  Hence, it suffices to show that $\I[\N]/\pair{\ms{i}_n \ge \ms{i}_{n+1}}_{n:\N}$ is indeed the sequential limit of the above $\I$-algebras,
  \[\begin{tikzcd}
    \cdots & {\I[n\!+\!1]/\ms{i}_0 \ge \cdots \ge \ms{i}_{n}} & {\I[n]/\ms{i}_0 \ge \cdots \ge \ms{i}_{n-1}} & \cdots \\
    & {\I[\N]/\pair{\ms{i}_n \ge \ms{i}_{n+1}}_{n:\N}}
    \arrow[from=1-1, to=1-2]
    \arrow["{i_{n} \mapsto 0}", from=1-2, to=1-3]
    \arrow[from=1-3, to=1-4]
    \arrow[dashed, from=2-2, to=1-1]
    \arrow["{f_{n+1}}"{description}, from=2-2, to=1-2]
    \arrow["{f_n}"{description}, from=2-2, to=1-3]
    \arrow[curve={height=15pt}, dashed, from=2-2, to=1-4]
  \end{tikzcd}\]
  where the map $f_n$ takes $\ms{i}_k$ to itself for $k\le n$, and takes $\ms{i}_k$ to $0$ for $k \ge n$. We can more directly compute the above sequential limit via \zcref{prop:simplicesasalgebra},
  \[\begin{tikzcd}
    \cdots & {\I[n\!+\!1]/\ms{i}_0 \ge \cdots \ge \ms{i}_{n}} & {\I[n]/\ms{i}_0 \ge \cdots \ge \ms{i}_{n-1}} & \cdots \\
    \cdots & {\Delta_{n+2}} & {\Delta_{n+1}} & \cdots
    \arrow[from=1-1, to=1-2]
    \arrow["{{\ms{i}_{n} \mapsto 0}}", from=1-2, to=1-3]
    \arrow["\cong", from=1-2, to=2-2]
    \arrow[from=1-3, to=1-4]
    \arrow["\cong", from=1-3, to=2-3]
    \arrow[from=2-1, to=2-2]
    \arrow["\pi_{n}", from=2-2, to=2-3]
    \arrow[from=2-3, to=2-4]
  \end{tikzcd}\]
  where $\pi_{n} : \Delta_{n+1} \to \Delta_{n+1}$ forgets the last entry, i.e.\ it takes $\ms{i}_0 \le \cdots \le \ms{i}_{n+1}$ to $\ms{i}_0 \le \cdots \le \ms{i}_{n}$. Hence, the sequential limit is given by
  \[ \lt_{n:\N} \Delta_{n} \cong \Delta_\infty\text{.} \]
  Now the desired result follows from \zcref{app:normalsigma}.
\end{proof}

\begin{remark}
  As indicated in \zcref{rem:qcreplete,rem:specarereplete}, both spectra and spatial $\I$-algebras are replete. The above result then implies that they are also chain complete, i.e.\ orthogonal to $\omega\hook\ov\omega$.
\end{remark}

\begin{corollary}[\AxiomNT, \AxiomSQCC]
  $\I^{\ov\omega} \cong \I^\omega \cong \Delta_\infty$. In particular, $\omega$ is not sober.
\end{corollary}
\begin{proof}
  By the above proof, we have $\I^{\ov\omega} \cong \I[\N]/\pair{\ms{i}_n \ge \ms{i}_{n+1}}_{n:\N} \cong \Delta_\infty$.
\end{proof}

\begin{remark}[$\sigma$-frames vs.\ bounded distributive lattices]\label{rem:whynotdis}
  Note that the inductivity of $\omega$ as shown in \zcref{prop:omegacolimit} is \emph{independent} of working with distributive lattices or $\sigma$-frames. The completeness result shown in \zcref{thm:complete} is the only one in this paper that works for $\sigma$-frames and not for bounded distributive lattices more generally. The reason is exactly because the normal form in \zcref{app:normalsigma} for the countably presented \emph{bounded distributive lattice} $\I[\N]/\pair{\ms{i}_n \ge \ms{i}_{n+1}}_{n:\N}$ will \emph{not} be isomorphic to $\Delta_\infty$. More specifically, for the countably presented bounded distributive lattice $\I[\N]/\pair{\ms{i}_n \ge \ms{i}_{n+1}}_{n:\N}$, since it is a sequential colimit of the following finitely presented bounded distributive lattices
  \[ 
  \begin{tikzcd}
    \cdots \ar[r] & \I[n]/\ms{i}_{0} \ge \cdots \ge \ms{i}_{n-1} \ar[r, hook] & \I[n\!+\!1]/\ms{i}_0 \ge \cdots \ge \ms{i}_n \ar[r, hook] & \cdots
  \end{tikzcd}
  \]
  and, due to the fact that the theory of bounded distributive lattices is \emph{finitary}, this sequential colimit of \emph{algebras} is computed the same as their \emph{underlying sets}. Via \zcref{prop:simplicesasalgebra}, the result can be identified as the subtype of $\Delta_\infty$,
  \[ \Delta_\omega \coloneq \scomp{i : \Delta_\infty}{\ex n\N \ist{i_n}}\text{,} \]
  where $\Delta_\omega \hook \Delta_\infty$ in some sense is the order-dual to the inclusion $\omega\hook\ov\omega$. The fact that we have \emph{all} infinite increasing sequences in the case for $\sigma$-frames is clearly due to the fact that we have countable disjunctions in $\sigma$-frames, as when identifying $p$ in $\I[\N]/\pair{\ms{i}_n \ge \ms{i}_{n+1}}_{n:\N}$ with $p = p_0 \vee \bigvee_{n:\N}p_{n+1}\wedge \ms{i}_n$.
\end{remark}

\begin{remark}[Limits of algebras induce locality for $\I$]\label{rem:limofalgloc}
  By inspecting the proof of \zcref{thm:complete}, it is clear that $\I$ is chain complete \emph{precisely} because there is a specific \emph{limiting diagram} of spatial $\I$-algebras. In general, any such limit diagram will induce an orthogonality property satisfied by the all spectra, and we will see many more examples in \zcref{sec:local}. From \zcref{rem:whynotdis} it is clear whether a diagram of $\I$-algebras is a limit heavily relies on the underlying algebraic theory. This is again a perfect example of how the algebraic properties of a theory has significant consequences for the internal logic of its classifying topos.
\end{remark}

\section{Local properties for the interval}\label{sec:local}

In this section we review some of the locality axioms we have introduced in \zcref{sec:locality}. As mentioned in the introduction, the observation is that even if we do not assume them to be true globally for the interval $\I$, we can still show the maps representing the locality axiom to be \emph{left orthogonal} to $\I$.  
Unlike chain completeness in \zcref{thm:complete}, the local conditions in \zcref{sec:locality} are all induced by limits of \emph{f.p.} $\I$-algebras. Thus in this section, it does not matter whether we work with bounded distributive lattices or $\sigma$-frames.

Starting from the simplest example, let us recall the local property (\AxiomNT). In general it is \emph{not} necessarily $\emp$, if we do not assume (\AxiomNT). However, we can look at the localisation class that thinks this map is invertible. The following fact shows that, from the perspective of $\I$, (\AxiomNT) indeed holds:

\begin{proposition}[\AxiomSQCI]\label{specisnontrivial}
  $\I$ is right orthogonal to $\emp \hook (0 = 1)$.
\end{proposition}
\begin{proof}
  Observe the proposition $0 = 1$ is propositionally equivalent to the (sober) spectrum $\spec(\I/0=1)$. By \zcref{prop:duality} we have:
  \[ 
    \I^{0=1} \cong 
    (\spec \I[\ms{i}])^{\spec{\I/0=1}}
    \cong 
    \I\alg(\I[\ms{i}], \I/0=1) 
    \cong 
    \mathbf{1}
  \]
  The last isomorphism is due to $\I/0=1$ being the terminal $\I$-algebra, since its underlying set is the singleton.
\end{proof}

As another example, consider the simplicial axiom (\AxiomSL). We skip the discussion of (\AxiomL) and (\AxiomCL), because (\AxiomSL) is strictly stronger, and it has a better-known geometric meaning. Recall from \zcref{sec:linearity-and-coskeletality} that (\AxiomSL) can be represented geometrically as an embedding $\Delta^2\cup\Delta_2 \hook \I^2$.

\begin{definition}[Triangulated type]
  A type $X$ is said to be \emph{triangulated} when it is right orthogonal to the inclusion $\Delta^2 \cup \Delta_2 \hook \I^2$.
\end{definition}

Again, (\AxiomSL) holds globally iff the above embedding is an equivalence, thus every type in this case will be triangulated. When it does not hold globally, we can still show:

\begin{proposition}[\AxiomSQCF]\label{specistriangulated}
  $\I$ is triangulated.
\end{proposition}
\begin{proof}
  To show that $\I$ is right orthogonal to $\Delta^2 \cup \Delta_2 \hook \I^2$, it is equivalent to verify that for any $f,g$ indicated below making the solid diagram below commute, there exists a unique lift $h$ making the whole diagram commute:
  \[\begin{tikzcd}
    \I \ar[r, "\delta"] \ar[d, "\delta"'] & \Delta^2 \ar[r, "f"] \ar[d] & \I \\
    \Delta_2 \ar[urr, "g"{description, near start}] \ar[r] & \I^2 \ar[ur, dashed, "h"']
    \arrow["\lrcorner"{anchor=center, pos=0.125}, draw=none, from=1-1, to=2-2]
  \end{tikzcd}\]
  Now by (\AxiomSQCF), since the vertices of the left square are all sober, equivalently it suffices to show we have a pullback of $\I$-algebras,
  \[\begin{tikzcd}
    {\I[\ms{i},\ms{j}]} \ar[r] \ar[d] & {\I[\ms{i},\ms{j}]/\ms{i} \ge \ms{j}} \ar[d, "{\ms{i},\ms{j} \mapsto\ms{k}}"] \\
    {\I[\ms{i},\ms{j}]/\ms{i} \le \ms{j}} \ar[r, "{\ms{i},\ms{j}\mapsto\ms{k}}"'] & {\I[\ms{k}]}
    \arrow["\lrcorner"{anchor=center, pos=0.125}, draw=none, from=1-1, to=2-2]
  \end{tikzcd}\]
  Recall from \zcref{prop:simplicesasalgebra} that the normal form of elements in $\I[\ms{i},\ms{j}]/\ms{i}\ge \ms{j}$. An element $p$ there can be viewed as the polynomial 
  \[ p = p_{0,0} \vee \ms{i}\wedge p_{1,0} \vee \ms{j}\wedge p_{1,1}\text{,} \]
  with $p_{0,0} \le p_{1,0} \le p_{1,1}$. And similarly for $q$ in $\I[\ms{i},\ms{j}]/\ms{i} \le \ms{j}$, 
  \[ q = q_{0,0} \vee \ms{j} \wedge q_{0,1} \vee \ms{i} \wedge q_{1,1}\text{.} \]
  They agree in $\I[\ms{k}]$ iff 
  \[ p_{0,0} = q_{0,0}, \quad p_{1,1} = q_{1,1}\text{.} \]
  This then gives us a polynomial in $\I[\ms{i},\ms{j}]$, which we write as follows,
  \[ (p_{0,0} \vee \ms{i} \wedge p_{1,0}) \vee \ms{j} \wedge (q_{0,1} \vee \ms{i} \wedge q_{1,1})\text{,} \]
  where $p_{0,0} = q_{0,0} \le q_{0,1}$ and $p_{1,0} \le p_{1,1} = q_{1,1}$. 
  By the general normal form for $A[\ms{i}]$ for any $\I$-algebra $A$, if we view $\I[\ms{i},\ms{j}]$ as $\I[\ms{i}][\ms{j}]$, the above exactly corresponds to the normal form of polynomials in $\I[\ms{i},\ms{j}]$; see also \citet[Ch.\ 1, Thm.\ 10.21]{lausch2000algebra}. This means the above is a pullback of $\I$-algebras.
\end{proof}

Finally, we discuss the even stronger locality condition (\AxiomOneCS). As mentioned in \zcref{sec:linearity-and-coskeletality}, the additional property of (\AxiomOneCS) is characterised by the embedding $\partial\Delta_2 \hook \Delta_2$.

\begin{definition}[1-coskeletal type]
  A type $X$ is said to be \emph{1-coskeletal} when it is both triangulated and right orthogonal to the boundary inclusion $\partial\Delta_2 \hook \Delta_2$.
\end{definition}

\begin{proposition}[\AxiomSQCF]\label{specis1t}
  $\I$ is 1-coskeletal.
\end{proposition}
\begin{proof}
  Completely similar to the proof of \zcref{specistriangulated}, it suffices to show that we have a limit diagram of $\I$-algebras as follows,
  \[\begin{tikzcd}
    & {\I[\ms{i},\ms{j}]/\ms{i}\le \ms{j}} \\
    \I[\ms{i}] & \I[\ms{k}] & \I[\ms{j}] \\
    \I & \I & \I
    \arrow["{\ms{j} \mapsto 1}"', from=1-2, to=2-1]
    \arrow["{\ms{i},\ms{j} \mapsto k}"{description}, from=1-2, to=2-2]
    \arrow["{\ms{i}\mapsto 0}", from=1-2, to=2-3]
    \arrow["\ms{i}\mapsto 1"', from=2-1, to=3-1]
    \arrow["\ms{i}\mapsto 0"{description, pos=0.75}, from=2-1, to=3-2]
    \arrow["\ms{k}\mapsto 1"{description, pos=0.75}, from=2-2, to=3-1]
    \arrow["\ms{k}\mapsto 0"{description, pos=0.75}, from=2-2, to=3-3]
    \arrow["\ms{j}\mapsto 1"{description, pos=0.75}, from=2-3, to=3-2]
    \arrow["\ms{j}\mapsto 0", from=2-3, to=3-3]
  \end{tikzcd}\]
  Now by the normalisation theorem, an element in the limit consists of $p$ in $\I[\ms{i}]$, $q$ in $\I[\ms{j}]$ and $r$ in $\I[\ms{k}]$, such that
  \[ p_1 = r_1, \quad p_0 = q_1, \quad r_0 = q_0\text{.} \]
  This exactly corresponds to a normal form in the algebra $\I[\ms{i},\ms{j}]/\ms{i}\le\ms{j}$ with $q_0 \le q_1 = p_0 \le p_1$ as follows (cf.\ \zcref{prop:simplicesasalgebra})
  \[ q_0 \vee \ms{j} \wedge q_1 \vee \ms{i} \wedge p_1\text{.} \]
  Hence, the above is a limit diagram of $\I$-algebras.
\end{proof}

Besides the locality principles discussed in \zcref{sec:locality}, we also consider the orthogonality classes emerging from \emph{synthetic category theory}, as introduced by \citet{riehl2017type}. 
A synthetic category will be a type that satisfies the internal orthogonality condition of being \emph{Segal complete}, \emph{Rezk complete}, and (according to some sources) \emph{simplicial}. Being simplicial is a sheaf condition that strengthens the triangulation property by stabilising the left class under pullback.

We start with the Segal completeness condition. Besides the outer horn discussed in \zcref{exm:hornsober}, we can also define the inner horn $\Lambda^2_1$ as a pushout,
\[
  \begin{tikzcd}
    1 \ar[d, "0"'] \ar[r, "1"] & \I \ar[d] \\
    \I \ar[r] & \Lambda^2_1
    \arrow["\lrcorner"{anchor=center, pos=0.125, rotate=180}, draw=none, from=2-2, to=1-1]
  \end{tikzcd}
\]
As a subtype of $\Delta_2$, it can be identified as $\scomp{(i,j) : \Delta_2}{\isf{i} \vee \ist{j}}$.

\begin{definition}[Segal complete types]
  $X$ is called \emph{Segal complete} when it is right orthogonal to $\Lambda^2_1 \hook \Delta_2$.
\end{definition}

\begin{remark}[Path transitivity]
  The Segal completeness condition has also been studied independently by \citet{fiore2001domains} under the name of \emph{path transitivity}.
\end{remark}

\begin{proposition}[\AxiomSQCP]
  $\I$ is Segal complete.
\end{proposition}
\begin{proof}
  By \zcref{prop:simplicesasalgebra} we have the following canonical isomorphism:
  \[ 
    \Delta^{n+1} \cong 
    \I[n]/i_1\leq\cdots\leq i_n
  \]
  
  In particular, we have $\Delta^3\cong \I[\ms{i},\ms{j}]/\ms{i}\leq\ms{j}$ and $\Delta^2 \cong \I[\ms{i}]$ and under Phoa's principle (SQCP), we have $\I^\I \cong \Delta^2\cong\I[\ms{i}]$. Therefore, to show that $\I$ is right orthogonal to the inner horn comparison map, it suffices to show that the following is a fibre product of $\I$-algebras:
  \[
  \begin{tikzcd}
    \I[\ms{i},\ms{j}]/\ms{i}\le\ms{j} \ar[r, "\ms{j} \mapsto 1"] \ar[d, "\ms{i} \mapsto 0"'] & \I[\ms{i}] \ar[d, "\ms{i}\mapsto 0"] \\
    \I[\ms{j}] \ar[r, "\ms{j} \mapsto 1"'] & \I
  \end{tikzcd}
  \]
  By the normal form theorem, an element in the pullback is given by $p$ in $\I[\ms{i}]$ and $q$ in $\I[\ms{j}]$ with $p_0 = q_1$. This way, it again corresponds to the following normal form in $\I[\ms{i},\ms{j}]/\ms{i} \le \ms{j}$ by \zcref{prop:simplicesasalgebra},
  \[ q_0 \vee \ms{j} \wedge q_1 \vee \ms{i} \wedge p_1\text{.} \]
  Hence, the above is again a pullback.
\end{proof}

Next we consider the Rezk completeness condition. Following \citet{buchholtz2021synthetic}, we can define the type $\mbb E$ classifying categorical equivalences as the colimit of the following diagram,
\[
\begin{tikzcd}
	& \I && \I && \I \\
	1 && {\Delta^2} && {\Delta^2} && 1
	\arrow[from=1-2, to=2-1]
	\arrow["{i \mapsto (i,i)}"{description}, from=1-2, to=2-3]
	\arrow["{i \mapsto (i,0)}"{description}, from=1-4, to=2-3]
	\arrow["{i \mapsto (1,i)}"{description}, from=1-4, to=2-5]
	\arrow["{i \mapsto (i,i)}"{description}, from=1-6, to=2-5]
	\arrow[from=1-6, to=2-7]
\end{tikzcd}
\]

\begin{definition}[Rezk types]
  We say that a type $X$ is \emph{Rezk complete} when it is right orthogonal to $\mbb E \to 1$.
\end{definition}

\begin{proposition}[\AxiomSQCP]
  $\I$ is Rezk complete.
\end{proposition}
\begin{proof}
  Notice by \zcref{thm:phoa-spectrum} $\I^\I$ classifies the canonical order on $\I$, which is antisymmetric by definition. Thus $\I$ will be Rezk complete.
\end{proof}

In fact, $\I$ is not only a synthetic category but in fact a synthetic \emph{poset}; this property too can be expressed in terms of orthogonality. To that end, we define the ``walking parallel pair'' $\I_{\rightrightarrows}$ to be the following pushout:
\[
\begin{tikzcd}
  2 \ar[r] \ar[d] & \I \ar[d] \\ 
  \I \ar[r] & \I_{\rightrightarrows}
  \arrow["\lrcorner"{anchor=center, pos=0.125, rotate=180}, draw=none, from=2-2, to=1-1]
\end{tikzcd}
\]

\begin{definition}[$\I$-separated types]
  $X$ is called \emph{$\I$-separated} when it is right orthogonal to $\I_{\rightrightarrows} \to \I$.
\end{definition}

Equivalently, $X$ is $\I$-separated iff $X^\I \to X \times X$ is an embedding. It is an immediate consequence of (\AxiomSQCP) that $\I$ is $\I$-separated.

The notion of synthetic categories and synthetic posets are formulated as these orthogonality classes:

\begin{definition}[Synthetic categories and synthetic posets]
  A type $X$ is a \emph{synthetic category}, if it is Segal and Rezk complete. We say it is a \emph{synthetic poset}, if it is also $\I$-separated.
\end{definition}

\begin{theorem}[\AxiomSQCF]\label{thm:soberposet}
  Any spectrum or spatial $\I$-algebra will be a synthetic poset.
\end{theorem}
\begin{proof}
  This follows from $\I$ being a synthetic poset, and spectra and spatial algebras are \emph{replete} as indicated in \zcref{rem:qcreplete,rem:specarereplete}.
\end{proof}

At the end of this section, we also discuss the example of $\omega$, which as we have seen is \emph{not} sober. We first observe its observational preorder again coincides with its expected pointwise order viewed as a subspace of $\I^\N$:

\begin{lemma}[\AxiomNT, \AxiomSQCC]\label{speconomegaiscan}
  The inclusion $\omega \hook \ov\omega$ is an order embedding for the observational preorder. In particular, the observational preorder on $\omega$ is induced pointwise as a subspace of $\ov\omega\hookrightarrow \I^\N$.
\end{lemma}
\begin{proof}
  The fact that $\omega\hook\ov\omega$ is an order embedding for the observational preorder follows from chain completeness of $\I$ in \zcref{thm:complete}. The fact that the observational preorder on $\ov\omega$ coincides with its satisfaction order as a spectrum follows from \zcref{lem:specorderofsober} by the fact that it is sober under (\AxiomSQCC).
\end{proof}

Furthermore, we can show that $\omega$ also satisfies a version of the Phoa principle, i.e.\ the path space $\omega^\I$ again classifies its pointwise order. For this we need to compute the path space $\omega^\I$.

Recall from \zcref{prop:omegacolimit} that, under (\AxiomNT) and (\AxiomSQCC), $\omega$ is a colimit $\omega \cong \ct_{n:\N}\Delta^n$. As indicated in \zcref{rem:whynotdis}, this does not depend on working with bounded distributive lattices or $\sigma$-frames. Type-theoretically, we have also shown that $\omega$ can be realised as the following subspace of $\ov\omega$,
\[ \omega \cong \scomp{i : \ov\omega}{\ex n\N \isf{i_n}}\text{.} \]
This way, the inclusion $\Delta^n \hook \omega$ can be viewed as follows, 
\[ \Delta^n \cong \scomp{i : \omega}{\isf{i_n}}\text{,} \]
which implies $\Delta^n \hook \omega$ is \emph{downward closed}. This allows us to directly compute $\omega^\I$:

\begin{proposition}[\AxiomNT, \AxiomSQCC]\label{lem:intervalcommuteomega}
  We have a family of equivalences
  \[ \omega^\I \cong \big({\ct_{n:\N}\Delta^n}\big)^\I \cong \ct_{n:\N}(\Delta^n)^\I\text{,} \]
  and similarly by replacing $\I$ with $\I^n$ or $\Delta^n$. In particular, $\omega^\I$ classifies the pointwise order on $\omega$.
\end{proposition}
\begin{proof}
  We show the following canonical map is an equivalence,
  \[ \ct_{n:\N}(\Delta^n)^\I \to \big({\ct_{n:\N}\Delta^n}\big)^\I\text{.} \]
  It is evident this map is an embedding, hence it suffices to show it is surjective. Given $f \colon \I \to \ct_{n:\N}\Delta^n$. By assumption, there merely exists $n:\N$ that $f(1)$ factors through $\Delta^n \hook \omega$. Now the claim is that the entire map $f$ factors through $\Delta^n$. This indeed holds, since we have shown in \zcref{speconomegaiscan} that the observational preorder on $\omega$ is pointwise. By monotonicity, for any $i:\I$ we have $fi \obsle f1$, which implies $fi$ belongs to $\Delta^n$ as well, since $\Delta^n \hook \omega$ is a downward closed. The same holds for cubes or simplices since they all have a top element. This way, $\omega^\I$ classifies the pointwise order on $\omega$ since all simplices $\Delta^n$ satisfy Phoa's principle as shown in \zcref{thm:phoa-spectrum}.
\end{proof}

As another consequence, we can also show the following general result establishing a large family of orthogonality conditions satisfied by $\omega$:

\begin{theorem}[\AxiomNT, \AxiomSQCC]\label{thm:omegaortho}
  Let $f \colon X \to Y$ be a map where $X,Y$ are finite colimits of cubes or simplices. If each $\Delta^n$ is $f$-local, then so is $\omega$.
\end{theorem}
\begin{proof}
  Let $Y \cong \ct_{r}Y_r$ be a finite colimit with each $Y_r$ a simplex or a cube.
  \[ \omega^Y \cong \lt_r\big({\ct_{n:\N}\Delta^n}\big)^{Y_r} \cong \lt_r\ct_{n:\N}(\Delta^n)^{Y_r} \cong \ct_{n:\N}\lt_r(\Delta^n)^{Y_r} \cong \ct_{n:\N}(\Delta^n)^Y \]
  The second equivalence holds by \zcref{lem:intervalcommuteomega}; the third holds since finite limits commutes with sequential colimits. Thus, if each $\Delta^n$ is $f$-local, i.e.\ $(\Delta^n)^Y \cong (\Delta^n)^X$, then so is $\omega$.
\end{proof}

\begin{remark}[Localisation classes closed under sequential colimits]
  Notice that in \zcref{lem:intervalcommuteomega} we do not need to use any form of choice principle to show the exponential $(-)^\I$ commutes with the sequential colimit $\ct_{n:\N}\Delta^n$, exactly because $\Delta^n \hook \omega$ is downward closed. However, for general sequential colimits, the same proof still goes through if we assume:
  \begin{axiom}[Choice principle for $\I$]
    For any type family $P$ over $\I$, we have
    \[ \prod_{i:\I}\pss{P(i)} \to \big\lVert\prod_{i:\I}P(i)\big\rVert\text{.} \]
  \end{axiom}
  Furthermore, the same holds for cubes and simplices since the types satisfying the choice principle are closed under finite products and retracts. Assuming $\I$ satisfies choice, following the proof of \zcref{thm:omegaortho}, one can show more generally that any orthogonality class specified by maps between finite colimits of simplices or cubes are always \emph{closed under sequential colimits}. However, not in every model of quasi-coherence $\I$ will satisfy the above choice principle. Hence, we do not include this result in the main text, as we tend to keep our assumptions as minimalistic as possible.
\end{remark}

\begin{figure}[h]
  \textbf{Quasi-coherence axioms:}
  \PrintAxiomSQCI
  \PrintAxiomSQCID
  \PrintAxiomSQCP
  \PrintAxiomSQCF
  \PrintAxiomSQCC

  \medskip
  \textbf{Locality axioms:}
  \PrintAxiomNT
  \PrintAxiomL
  \PrintAxiomCL
  \PrintAxiomSL
  \PrintAxiomOneCS
  \vspace{6ex}

  \caption{Summary of axioms considered in this paper. Unlike the usual axiom sets for synthetic domain theory~\citep{hyland1990first}, the goal of these axioms is not to modularly specify all the properties of the interval with minimal overlap but instead to identify a sequence of increasingly strong descriptions of intervals that may play a role in synthetic domain theory.}
\end{figure}

\section{New models for synthetic domain theory}\label{sec:model}

It is now instructive to discuss models for the axioms we have used in the previous sections. For this purpose, we now work externally under the assumption that the base topos $\Set$ satisfies the axiom of choice. Our axioms for synthetic domain theory can be organised into two classes:

\begin{enumerate}
  \item The quasi-coherence principles (\AxiomSQCF), (\AxiomSQCC);
  \item Various locality axioms discussed in \zcref{sec:locality}.
\end{enumerate}

For any Horn theory $\T$, we know from \citet{blechschmidt2020general,blechschmidt2021using} that the generic $\T$-model $U_\T\in\Set[\T]$ satisfies quasi-coherence for finitely presented $U_\T$-algebras, viz.\ (\AxiomSQCF). As mentioned in \zcref{subsec:qc}, the proof of quasi-coherence for finitely presented $U_\T$-algebras adapts readily to the countably presented case, provided we work with the larger site
\[ \Set[\T]_\omega \coloneq [\mmod\T\cp,\Set]\text{,} \]
where $\mmod\T\cp$ is the category of countably presented $\T$-models. The generic $\T$-model in $\Set[\T]_\omega$ is again defined to be the forgetful functor $\mmod\T\cp \to \Set$, and we also denote it as $U_\T$. In other words, (\AxiomSQCC) holds in the larger topos $\Set[\T]_\omega$. And in this case, we can also allow $\T$ to contain algebraic operations of countable arity. 

These two topoi $\Set[\T]$ and $\Set[\T]_\omega$ are intimately related. There is a fully faithful and left exact inclusion of sites
\[ \mmod\T\fp\op \hook \mmod\T\cp\op \]
inducing an adjoint triple (cf.\ \citet[Thm.\ 7.20]{caramello2019denseness}),
\[\begin{tikzcd}
  {\Set[\T]_\omega} & {\Set[\T]}
  \arrow[""{name=0, anchor=center, inner sep=0}, "\Gamma"{description}, from=1-1, to=1-2]
  \arrow[""{name=1, anchor=center, inner sep=0}, "\Delta"', curve={height=18pt}, from=1-2, to=1-1]
  \arrow[""{name=2, anchor=center, inner sep=0}, "\nabla", curve={height=-18pt}, from=1-2, to=1-1]
  \arrow["\dashv"{anchor=center, rotate=-90}, draw=none, from=0, to=2]
  \arrow["\dashv"{anchor=center, rotate=-90}, draw=none, from=1, to=0]
\end{tikzcd}\]
equivalently, a \emph{local geometric morphism} $\Set[\T]_\omega \surj \Set[\T]$.

More generally, it is already observed by \citet[Thm.\ 4.11]{blechschmidt2020general} that if we have a topology $J$ on $\mmod\T\fp\op$ where $U_\T$ is a $J$-sheaf, then (\AxiomSQCF) again holds in the sheaf subtopos $\sh(\mmod\T\fp\op,J)$. We refer to such a topology $J$ as \emph{$\T$-admissible}, or simply \emph{admissible} when no confusion could arise. 
For instance, since $U_\T$ is representable, any subcanonical topology will in particular be admissible. In the same way we can define admissible topology $J$ for the larger presheaf topos $\mmod\T\cp\op$, which again is a topology making the generic model $U_\T$ a sheaf. In this case, the sheaf subtopos $\sh(\mmod\T\cp\op,J)$ also models (\AxiomSQCC). 
As an example, this is the theoretical basis of the quasi-coherence principle for countably presented Boolean algebras in the topos of light condensed sets introduced by Clausen and Scholze, as shown by \citet{cherubini2024foundation}.

An admissible topology $J$ on $\mmod\T\cp\op$ restricts to one on $\mmod\T\fp\op$. In this case, the adjoint triple mentioned above between the two presheaf topoi will restrict to the sheaf subtopoi,
\[\begin{tikzcd}
  {\sh(\mmod\T\cp\op,J)} & {\sh(\mmod\T\fp\op,J)}
  \arrow[""{name=0, anchor=center, inner sep=0}, "\Gamma"{description}, from=1-1, to=1-2]
  \arrow[""{name=1, anchor=center, inner sep=0}, "\Delta"', curve={height=18pt}, from=1-2, to=1-1]
  \arrow[""{name=2, anchor=center, inner sep=0}, "\nabla", curve={height=-18pt}, from=1-2, to=1-1]
  \arrow["\dashv"{anchor=center, rotate=-90}, draw=none, from=0, to=2]
  \arrow["\dashv"{anchor=center, rotate=-90}, draw=none, from=1, to=0]
\end{tikzcd}\]
which again identifies $\sh(\mmod\T\cp\op,J)$ as a local topos over $\sh(\mmod\T\fp\op,J)$. These admissible topologies on $\mmod\T\fp$ or $\mmod\T\cp$ are exactly the required data to validate various local properties of the generic model $U_\T$. 

More specifically for us, let $\sFrm$ be the category of $\sigma$-frames, i.e.\ whose objects are posets with finite meets and countable joins, where binary meets distribute over countable joins. We will also use $\sFrm\cp$ to denote the full subcategory spanned by countably presented $\sigma$-frames, and $\sFrm\fp$ to denote the finitely presented ones. Since finitely presented $\sigma$-frames are simply finitely presented bounded distributive lattices, we have an isomorphism $\sFrm\fp \cong \DL\fp$.

It is well-known the dual category of $\DL\fp \cong \sFrm\fp$ is the category $\Pos\fp$ of finite posets. In this case, we can also have a fairly explicit description of the larger dual category of $\sFrm\cp$. The first observation is that any countably presented $\sigma$-frame $A$ is indeed a \emph{frame}, i.e.\ it has \emph{all} joins and finite meets, which distributes with each other. This is clear for all finitely presented $\sigma$-frames, since they are simply finite bounded distributive lattices. To see this for the countable case, consider the countably generated free $\sigma$-frame $\Sigma[\N]$:

\begin{lemma}\label{lem:cgfreesframe}
  We have an isomorphism of $\sigma$-frames
  \[ \Sigma[\N] \cong \Pos(P_f(\N),\Sigma)\text{,} \]
  where $P_f(\N)$ is the poset of finite subsets of $\N$.
\end{lemma}
\begin{proof}
  The free $\sigma$-frame can be generated by first freely adding finite meets to the discrete poset $\N$, and then freely adding all countable joins. The first step results in the poset $P_f(\N)\op$. Now since $P_f(\N)\op$ is countable, freely adding all countable joins is equivalently freely adding \emph{all} joins, which is achieved by the presheaf construction. This way,
  \[ \Sigma[\N] \cong \Pos((P_f(\N)\op)\op,\Sigma) \cong \Pos(P_f(\N),\Sigma). \qedhere \]
\end{proof}

\begin{corollary}\label{cor:dualsframe}
  There is a fully faithful embedding
  \[ \sFrm\cp \hook \Frm\text{,} \]
  preserving all countable colimits. This is again fully faithful when composed with $\ms{pt} \colon \Frm\op \to \Topp$, where $\Topp$ is the category of topological spaces.
\end{corollary}
\begin{proof}
  Any countably presented $\sigma$-frame will be isomorphic to one of the form $\Sigma[\N]/R$ for some countably generated congruence $R$. By \zcref{lem:cgfreesframe}, $\Sigma[\N]$ is a frame, so is any of its quotient. By Theorem 6.2.4 of \citet{makkai2006first}, such frames are indeed \emph{spatial}, hence they fully faithfully embed into the category of topological spaces via the functor $\pt$.
\end{proof}

This way, we can view $\sFrm\cp\op$ as a certain class of topological spaces. Notice that since all countably presented $\sigma$-frames will be quotients of $\Sigma[\N]$, their dual spaces will be a subspace of $\pt(\Sigma[\N])$, which we can compute quite easily:

\begin{lemma}
  The space of points of $\Sigma[\N]$ is Scott's graph model $G$, which is the countable product of the Sierpi\'nski space $G \cong \Sigma^\N$.
\end{lemma}
\begin{proof}
  Note $\pt$ takes colimits in $\Frm$ to limits in $\Topp$, since $\pt \colon \Loc \to \Topp$ is a right adjoint. Since $\Sigma[\N]$ is the countable coproduct of the free frame on one generator, we have
  \[ \pt(\Sigma[\N]) \cong \pt(\Sigma[\ms{x}])^\N \cong \Sigma^\N\text{,} \]
  Here $\pt(\Sigma[\ms{x}]) \cong \Sigma$ follows from a simple computation.
\end{proof}

This way, if we write $\wTop$ as the essential image of the fully faithful functor $\pt \colon \sFrm\cp\op \to \Topp$, its objects will all be subspaces of $G$. We would then have the following diagram,
\[
\begin{tikzcd}
  \sFrm\fp\op \ar[d, hook] \ar[r, "\simeq"] & \Pos\fp \ar[d, hook] \\
  \sFrm\cp\op \ar[r, "\simeq"'] & \wTop
\end{tikzcd}
\]
where here the inclusion $\Pos\fp \hook \wTop$ simply takes each finite poset to its Alexandroff topological space. Hence, at the presheaf level, we have a local geometric morphism
\[ \psh(\wTop) \surj \psh(\Pos\fp)\text{.} \]
In this case, the representable presheaf on the Sierpinski space will again be a $\sigma$-frame $\I$ in the presheaf $\psh(\wTop)$.

Below we discuss the corresponding admissible topologies on $\wTop$, modelling the various locality principles we have considered in \zcref{sec:locality}. We encourage the readers to notice the connection between the topologies we discuss below, and the developments in \zcref{sec:local}.

\begin{example}[\AxiomNT]
  To model (\AxiomNT), we would want the empty sieve on $\pt(\Sigma/0=1) \cong \emp$ to be a covering. Since $\emp$ is a strict initial object in $\wTop$, this is a subcanonical topology. Thus, this gives us a least topology $J_{\ms{NT}}$ that models (\AxiomNT), and we have
  \[ \sh(\wTop,J_{\ms{NT}}) \cong \psh(\wTop_+)\text{,} \]
  where $\wTop_+$ is the full subcategory of $\wTop$ excluding $\emp$. The induced local geometric morphism now is given by 
  \[ \Gamma \colon \psh(\wTop_+) \surj \psh(\Pos_{\mr{f.p.,+}})\text{,} \]
  where, again $\Pos_{\mr{f.p.,+}}$ is the full subcategory of $\Pos_{\mr{f.p.}}$ consisting of non-empty posets, and $\psh(\Pos_{\mr{f.,+}})$ is the classifying topos for bounded distributive lattices that are non-trivial.
\end{example}

\begin{example}[\AxiomL]
  The additional axiom for (\AxiomL) besides (\AxiomNT) is that 
  \[ x \vee y = 1 \vdash x = 1 \vee y = 1\text{,} \] 
  This means the dual embeddings of the following quotients need to form a covering family,
  \[ \Sigma[\ms{y}] \cong \Sigma[\ms{x},\ms{y}]/\ms{x} \twoheadleftarrow \Sigma[\ms{x},\ms{y}]/\ms{x}\vee \ms{y} \surj \Sigma[\ms{x},\ms{y}]/\ms{y} \cong \Sigma[\ms{x}]\text{.} \]
  It is easy to see that $\pt(\Sigma[\ms{x},\ms{y}]/\ms{x} \vee \ms{y})$ is the space $\Lambda^2_1$ obtained by glueing the open points of two copies of Sierpinski spaces together, i.e.\ it is the following pushout:
  \[
  \begin{tikzcd}
    1 \ar[d, "1"'] \ar[r, "1"] & \Sigma \ar[d, "l"] \\ 
    \Sigma \ar[r, "r"'] & \Lambda^2_1
    \arrow["\lrcorner"{anchor=center, pos=0.125, rotate=180}, draw=none, from=2-2, to=1-1]    
  \end{tikzcd}
  \]
  The least topology $J_{\ms L}$ is thus generated by the empty covering on $\emp$, and the covering $\set{l,r : \Sigma \to \Lambda^2_1}$. Since we have the pushout above, this topology is again subcanonical. 
  
  In this case, it is not hard to give an explicit description of a covering family in $\Pos\fp$: a family of maps is a $\ms J_{\ms L}$-covering on $P \in \Pos\fp$ iff it contains a (finite) subfamily of \emph{upward closed subsets} $\set{P_i\subseteq P}_{i\in I}$, where $P \cong \bigcup_{i\in I}P_i$; see a similar calculation for commutative rings in \citet[VIII. 6]{maclane1992sheaves}. Following the terminology for algebraic geometry, this can be denoted as the \emph{Zariski topology}, and the local geometric morphism
  \[ \sh(\wTop,J_{\ms L}) \surj \sh(\Pos\fp,J_{\ms L}) \simeq \mb{Zar}(\mbb D)\text{,} \]
  is over the \emph{Zariski topos} $\mb{Zar}(\mbb D)$ for bounded distributive lattices.
\end{example}

\begin{example}[\AxiomSL]
  An even stronger axiom is the linearity axiom, 
  \[ \top \vdash x \le y \vee y \le x\text{,} \]
  which requires the dual embeddings of the following quotients to form a covering family,
  \[ \Sigma[\ms{x},\ms{y}]/\ms{x} \le \ms{y} \twoheadleftarrow \Sigma[\ms{x},\ms{y}] \surj \Sigma[\ms{x},\ms{y}]/\ms{y} \le \ms{x}\text{.} \]
  We can compute that $\spec \Sigma[\ms{x},\ms{y}] \cong \Sigma^2$, and 
  \[ \pt(\Sigma[\ms{x},\ms{y}]/\ms{x} \le \ms{y}) \cong \Sigma_2 \cong \set{0 < 1 < 2}\text{.} \]
  Again we have a pushout diagram,
  \[
  \begin{tikzcd}
    \Sigma \ar[r, "\partial"] \ar[d, "\partial"'] & \Sigma_2 \ar[d, "l"] \\ 
    \Sigma_2 \ar[r, "r"'] & \Sigma^2
    \arrow["\lrcorner"{anchor=center, pos=0.125, rotate=180}, draw=none, from=2-2, to=1-1]
  \end{tikzcd}
  \]
  where $\partial : \Sigma \to \Sigma_2$ maps $\Sigma$ to the end points of $\Sigma_2$, and $l,r : \Sigma_2 \hook \Sigma^2$ is the two $\Sigma_2$-chains in $\Sigma^2$. The least topology $\ms J_{\ms{SL}}$ for the linearity axiom is thus generated by the empty covering on $\emp$, and $\set{l,r : \Sigma_2 \hook \Sigma^2}$. Similarly, given this pushout, $J_{\ms{SL}}$ will be subcanonical.

  The topos $\sh(\wTop,J_{\ms{SL}})$ is closely related to simplicial sets. This is due to Joyal's result that the category of simplicial sets is the classifying topos of strict bounded linear orders, and thus we have an equivalence 
  \[ \sh(\Pos\fp,J_{\ms{SL}}) \cong \psh(\Delta)\text{.} \]
  This makes $\sh(\wTop,J_{\ms{SL}}) \surj \psh(\Delta)$ a local topos over simplicial sets.
\end{example}

\begin{example}[\AxiomOneCS]\label{exm:model1T}
  Similarly to the case of (\AxiomSL), the 1-truncation axiom requires dual embeddings of the following quotient maps
  \[\begin{tikzcd}
    & {\Sigma[\ms{x},\ms{y}]/\ms{x}\le\ms{y}} \\
    {\Sigma[\ms{x},\ms{y}]/\ms{x}=0} & {\Sigma[\ms{x},\ms{y}]/\ms{x}=\ms{y}} & {\Sigma[\ms{x},\ms{y}]/\ms{y}=1}
    \arrow[two heads, from=1-2, to=2-1]
    \arrow[two heads, from=1-2, to=2-2]
    \arrow[two heads, from=1-2, to=2-3]
  \end{tikzcd}\]
  to be a covering, which translates to the fact that the three inclusions
  \[ d_0,d_1,d_2 : \Sigma \hook \Sigma_2 \]
  covers $\Sigma_2$. Again we have a colimit diagram as follows, 
  \[\begin{tikzcd}
    1 & 1 & 1 \\
    \Sigma & \Sigma & \Sigma \\
    & {\Sigma_2}
    \arrow["0"', from=1-1, to=2-1]
    \arrow["1"{description, pos=0.25}, from=1-1, to=2-2]
    \arrow["0"{description, pos=0.25}, from=1-2, to=2-1]
    \arrow["1"{description, pos=0.25}, from=1-2, to=2-3]
    \arrow["0"{description, pos=0.25}, from=1-3, to=2-2]
    \arrow["0", from=1-3, to=2-3]
    \arrow["{d_0}"', from=2-1, to=3-2]
    \arrow["{d_1}"{description}, from=2-2, to=3-2]
    \arrow["{d_2}", from=2-3, to=3-2]
  \end{tikzcd}\]
  which implies $J_{\ms{1cS}}$ is subcanonical.
  
  Similar to the case of (\AxiomSL), $\sh(\wTop,J_{\ms{1cS}})$ is closely related to the topos of \emph{truncated simplicial sets} $\psh(\Delta_{\le 1})$. This is again induced by the fact that the classifying topos for 1-coskeletal bounded distributive lattices is given by 
  \[ \sh(\Pos\fp,J_{\ms{1cS}}) \cong \psh(\Delta_{\le 1})\text{.} \]
  Hence, this again gives us a local geometric morphism
  \[ \sh(\wTop,J_{\ms{1cS}}) \surj \psh(\Delta_{\le 1})\text{.} \]
\end{example}

\section{Future directions}

\subsection{Domain theory with quasi-coherence}

In this paper we have explained the axioms of synthetic domain theory using the quasi-coherence principle for bounded distributive lattices and $\sigma$-frames. As we have seen from various examples (\zcref{prop:liftingofalgebra,prop:liftofsober}, \zcref{exm:2issober} etc.) the quasi-coherence principle moreover helps with the actual \emph{computation} for operations on domains. Thus, the natural step next is to further develop domain theory within this framework.


\subsection{Relationship with Taylor's \emph{Abstract Stone Duality}}

We would be remiss in failing to mention Taylor's framework of \emph{Abstract Stone Duality}~\citep{Taylor2011,TaylorP:insema,TaylorP:sobsc}, which is based on a similar duality between ``algebras'' and ``spaces''. One apparent difference in methodology is that in Taylor's case, the duality is of the form $\Sigma^{(-)} \dashv \Sigma^{(-)}$ and is moreover required to be monadic. In our setting, the spectrum of an algebra consists not in \emph{all} functions from the algebra into the dualising object but rather only the ones that preserve the structure of observational algebras—the \emph{primes} in Taylor's terminology.

Our own adjunction $\opens\dashv\spec$ shall not be monadic, but a natural step toward understanding the relationship between abstract Stone duality and synthetic domain theory in the language of quasi-coherence might be to restrict $\opens\dashv\spec$ to an adjunction involving certain more well-behaved subuniverses of $\Set$. We leave this investigation for future work.

\subsection{Quasi-coherence for related synthetic mathematics}

Both synthetic topology~\citep{bauer2009dedekind} and synthetic computability theory~\citep{RN552} involve similar structures as synthetic domain theory. In particular, an interval object $\I$ consisting of a dominance seems to be crucial in all three cases. We have showcased in \zcref{sec:dominance} that quasi-coherence for a wide variety of theories will produce such a structure, and it will be interesting to see connections with quasi-coherence to these two types of synthetic mathematics as well.

\subsection{Connection with existing models}\label{subsec:compare}

Finally, it will be instructive to compare the sheaf models for synthetic domain theory constructed by \citet{FIORE1997151} with the models discussed in \zcref{sec:model}. Recall the topos $\mc H$ constructed in \emph{loc.\ cit.}\ is the following sheaf topos,
\[ \mc H \coloneq \sh(\mb P,J_{\ms{can}})\text{,} \]
where $\mb P$ be the full subcategory of the category $\wCPO$ of $\omega$-cpos consisting of retracts of the Scott's graph model $G$, and $J_{\ms{can}}$ is the canonical topology on $\mb P$. An $\omega$-cpo is simply a poset having joins of countable chains, and morphisms between them are maps preserving joins of countable chains.

Here we observe that the inclusion $\mb P \hook \wTop$ is \emph{fully faithful}, because
\begin{align*}
  \wCPO(G,G)
  &\cong \Pos(P_f(\N),P(\N)) \\
  &\cong \Pos(\N,\Pos(P_f(\N),\Sigma)) \\
  &\cong \sFrm\cp(\Pos(P_f(\N),\Sigma),\Pos(P_f(\N),\Sigma)) \\ 
  &\cong \wTop(G,G)
\end{align*}
The first isomorphism holds because $G$ is the \emph{free} $\omega$-cpo on $P_f(\N)$; the second holds by formal manipulation; the third is due to the fact that $\Pos(P_f(\N),\Sigma)$ is the free $\sigma$-frame over the discrete poset $\N$; the final one is given by the duality between countably presented $\sigma$-frames and $\wTop$. 

This leads us to the question of the exact relationship between $\mc H$ and the sheaf model $\sh(\wTop,J_{\mr{can}})$ we have constructed. A detailed investigation requires us to better understand the categorical properties of $\wTop$, and we leave this for future work.

\subsection*{Acknowledgements}

We are grateful to Marcelo Fiore for alerting us to the connection between synthetic quasi-coherence and Blass's observations concerning functions on universal algebras~\citep{RN879}. We are also thankful to Jem Lord for suggestion some improvements on a lemma and a previously made conjecture.
 This work was funded by the United States Air Force Office of Scientific Research under grant FA9550-23-1-0728 (\emph{New Spaces for Denotational Semantics}; Dr.\ Tristan Nguyen, Program Manager). Views and opinions expressed are however those of the author only and do not necessarily reflect those of AFOSR.

\bibliographystyle{apalike} 
\bibliography{mybib}

@book{BarrMichael1985Ttat,
series = {Grundlehren der mathematischen Wissenschaften; 278},
publisher = {Springer},
booktitle = {Toposes, triples, and theories},
isbn = {0387961151},
year = {1985},
title = {Toposes, triples, and theories},
language = {eng},
address = {New York},
author = {Michael Barr and Charles Wells},
}

@article{kelly1993adjunctions,
  title={Adjunctions whose counits are coequalizers, and presentations of finitary enriched monads},
  author={Kelly, G Maxwell and Power, A John},
  journal={Journal of pure and applied algebra},
  volume={89},
  number={1-2},
  pages={163--179},
  year={1993},
  publisher={Elsevier}
}

@article{fiore2001domains,
  title={{Domains in H}},
  author={Fiore, Marcelo P. and Rosolini, Giuseppe},
  journal={Theoretical Computer Science},
  volume={264},
  number={2},
  pages={171--193},
  year={2001},
  publisher={Elsevier}
}

@article{RN552,
   author = {Bauer, Andrej},
   title = {First steps in synthetic computability theory},
   journal = {Electronic Notes in Theoretical Computer Science},
   volume = {155},
   pages = {5-31},
   ISSN = {1571-0661},
   year = {2006},
   type = {Journal Article}
}

@article{bauer2009dedekind,
  title={{The Dedekind reals in abstract Stone duality}},
  author={Bauer, Andrej and Taylor, Paul},
  journal={Mathematical Structures in Computer Science},
  volume={19},
  number={4},
  pages={757--838},
  year={2009},
  publisher={Cambridge University Press}
}

@book{makkai2006first,
  title={First order categorical logic: model-theoretical methods in the theory of topoi and related categories},
  author={Makkai, Michael and Reyes, Gonzalo E},
  volume={611},
  year={2006},
  publisher={Springer}
}

@book{PhoaWesleyKym-Son1991DtiR,
publisher = {Ph.D. Dissertation},
year = {1991},
title = {Domain theory in Realizability Toposes.},
language = {eng},
address = {University of Cambridge},
author = {Phoa, Wesley Kym-Son},
}

@article{RN879,
   author = {Blass, Andreas},
   title = {Functions on universal algebras},
   journal = {Journal of Pure and Applied Algebra},
   volume = {42},
   number = {1},
   pages = {25-28},
   ISSN = {0022-4049},
   year = {1986},
   type = {Journal Article}
}

@article{buchholtz2021synthetic,
  title={Synthetic fibered $(\infty, 1)$-category theory},
  author={Buchholtz, Ulrik and Weinberger, Jonathan},
  year={2023},
  journal={Higher Structures},
  pages={74-165},
  volume={7}
}

@article{riehl2017type,
  title={A type theory for synthetic $\infty$-categories},
  author={Riehl, Emily and Shulman, Michael},
  journal={Higher Structures},
  year={2017},
  pages={147-224},
  volume={1}
}

@article{JIBLADZE1997185,
title = {A presentation of the initial lift-algebra},
journal = {Journal of Pure and Applied Algebra},
volume = {116},
number = {1},
pages = {185-198},
year = {1997},
issn = {0022-4049},
doi = {https://doi.org/10.1016/S0022-4049(96)00108-9},
url = {https://www.sciencedirect.com/science/article/pii/S0022404996001089},
author = {Mamuka Jibladze},
}

@article{cherubini2024foundation,
  title={A Foundation for Synthetic {Stone} Duality},
  author={Cherubini, Felix and Coquand, Thierry and Geerligs, Freek and Moeneclaey, Hugo},
  journal={arXiv preprint arXiv:2412.03203},
  year={2024}
}

@article{gratzer2024directed,
  title={Directed univalence in simplicial homotopy type theory},
  author={Gratzer, Daniel and Weinberger, Jonathan and Buchholtz, Ulrik},
  journal={arXiv preprint arXiv:2407.09146},
  year={2024}
}

@article{Cherubini_Coquand_Hutzler_2024, title={A foundation for synthetic algebraic geometry}, volume={34}, DOI={10.1017/S0960129524000239}, number={9}, journal={Mathematical Structures in Computer Science}, author={Cherubini, Felix and Coquand, Thierry and Hutzler, Matthias}, year={2024}, pages={1008-1053}}

@misc{blechschmidt2020general,
  title={A general nullstellensatz for generalized spaces},
  author={Blechschmidt, Ingo},
  year={2020},
  note={Unpublished note.}
}

@article{VANOOSTEN2000233,
title = {Axioms and (counter) examples in synthetic domain theory},
journal = {Annals of Pure and Applied Logic},
volume = {104},
number = {1},
pages = {233-278},
year = {2000},
issn = {0168-0072},
doi = {https://doi.org/10.1016/S0168-0072(00)00014-2},
url = {https://www.sciencedirect.com/science/article/pii/S0168007200000142},
author = {Jaap {van Oosten} and Alex K. Simpson},
}

@book{lausch2000algebra,
  title={Algebra of polynomials},
  author={Lausch, Hans and Nobauer, Wilfred},
  year={2000},
  publisher={Elsevier}
}

@inproceedings{hyland1990first,
  title={First steps in synthetic domain theory},
  author={Hyland, J Martin E},
  booktitle={Category Theory: Proceedings of the International Conference held in Como, Italy, July 22--28, 1990},
  pages={131--156},
  year={1990},
  organization={Springer}
}

@article{FIORE1997151,
  title = {Two models of synthetic domain theory},
  journal = {Journal of Pure and Applied Algebra},
  volume = {116},
  number = {1},
  pages = {151-162},
  year = {1997},
  issn = {0022-4049},
  doi = {https://doi.org/10.1016/S0022-4049(96)00164-8},
  url = {https://www.sciencedirect.com/science/article/pii/S0022404996001648},
  author = {Marcelo P. Fiore and Giuseppe Rosolini}
}

@article{blechschmidt2021using,
  title={Using the internal language of toposes in algebraic geometry},
  author={Blechschmidt, Ingo},
  journal={arXiv preprint arXiv:2111.03685},
  year={2021}
}

@phdthesis{rosolini1986continuity,
  title={Continuity and Effectiveness in Topoi},
  author={Rosolini, Giuseppe},
  year={1986},
  school={University of Oxford}
}

@Book{hottbook,
  author =    {{The Univalent Foundations Program}},
  title =     {Homotopy Type Theory: Univalent Foundations of Mathematics},
  publisher = {\textsf{https://homotopytypetheory.org}},
  address =   {Institute for Advanced Study},
  year =      {2013}
}

@book{maclane1992sheaves,
  title={Sheaves in geometry and logic: A first introduction to topos theory},
  author={MacLane, Saunders and Moerdijk, Ieke},
  year={1992},
  publisher={Springer Science \& Business Media}
}

@article{caramello2019denseness,
  title={Denseness conditions, morphisms and equivalences of toposes},
  author={Caramello, Olivia},
  journal={arXiv preprint arXiv:1906.08737},
  year={2019}
}

@book{adamek1994locally,
  title={Locally presentable and accessible categories},
  author={Ad{\'a}mek, Ji{\v{r}}{\'\i} and Rosicky, Ji{\v{r}}{\'\i}},
  volume={189},
  year={1994},
  publisher={Cambridge University Press}
}

@book{johnstone2002sketches,
  title={Sketches of an Elephant: A Topos Theory Compendium},
  author={Johnstone, Peter T},
  volume={1, 2},
  year={2002},
  publisher={Oxford University Press}
}

@inproceedings{fiore-plotkin:1996,
  author = {Fiore, Marcelo P. and Plotkin, Gordon D.},
  editor = {van Dalen, Dirk and Bezem, Marc},
  publisher = {Springer},
  booktitle = {Computer Science Logic, 10th International Workshop, {CSL} '96, Annual Conference of the EACSL, Utrecht, The Netherlands, September 21-27, 1996, Selected Papers},
  year = {1996},
  doi = {10.1007/3-540-63172-0\_36},
  pages = {129--149},
  series = {Lecture Notes in Computer Science},
  title = {An Extension of Models of Axiomatic Domain Theory to Models of Synthetic Domain Theory},
  volume = {1258},
}

@incollection{Taylor2011,
author="Taylor, Paul",
editor="Sommaruga, Giovanni",
title="Foundations for Computable Topology",
bookTitle="Foundational Theories of Classical and Constructive Mathematics",
year="2011",
publisher="Springer Netherlands",
address="Dordrecht",
pages="265--310",
isbn="978-94-007-0431-2",
doi="10.1007/978-94-007-0431-2_14",
url="https://doi.org/10.1007/978-94-007-0431-2_14"
}

@article{reus-streicher:1999,
  author = {Reus, Bernhard and Streicher, Thomas},
  publisher = {Cambridge University Press},
  year = {1999},
  doi = {10.1017/S096012959900273X},
  journal = {Mathematical Structures in Computer Science},
  number = {2},
  pages = {177--223},
  title = {General synthetic domain theory --- a logical approach},
  volume = {9},
}

@article{simpson:2004,
  author = {Simpson, Alex},
  year = {2004},
  doi = {10.1016/j.apal.2003.12.005},
  issn = {0168-0072},
  journal = {Annals of Pure and Applied Logic},
  note = {Papers presented at the 2002 IEEE Symposium on Logic in Computer Science (LICS)},
  number = {1},
  pages = {207--275},
  title = {Computational adequacy for recursive types in models of intuitionistic set theory},
  volume = {130},
}

@article{TaylorP:insema,
      author    = {Taylor, Paul},
      title     = {Inside every model of {Abstract Stone Duality}
                   lies an {Arithmetic Universe}},
      journal   = {Electronic Notes in Theoretical Computer Science},
      publisher = {Elsevier},
      year      = 2005,
      volume    = 122, pages = {247--296},
      url       = {PaulTaylor.EU/ASD/insema}}

@article{TaylorP:sobsc,
    author    = {Taylor, Paul},
    title     = {Sober Spaces and Continuations},
    journal   = {Theory and Applications of Categories},
    publisher = {Mount Allison University},
    year      = 2002, month = {July},
    volume    = 10, number = 12, pages = {248--299},
    url       = {PaulTaylor.EU/ASD/sobsc},
    amsclass  = {06D22, 06E15, 18B30, 18C50,
                  22A26, 54C35, 54D10, 54D45}}

\end{document}